\pdfoutput=1

\documentclass{lmcs}
\usepackage{macros}

\theoremstyle{definition}
\newtheorem{con}[thm]{Construction}
\newtheorem{exegesis}[thm]{Exegesis}

\Crefname{con}{Construction}{Constructions}
\Crefname{defi}{Definition}{Definitions}
\Crefname{fact}{Fact}{Facts}
\Crefname{thm}{Theorem}{Theorems}
\Crefname{exa}{Example}{Examples}
\Crefname{rem}{Remark}{Remarks}
\Crefname{lem}{Lemma}{Lemmas}
\Crefname{cor}{Corollary}{Corollaries}
\Crefname{diagram}{Diagram}{Diagrams}

\begin{document}

\title{Denotational semantics of general store and polymorphism}

\author[J.~Sterling]{Jonathan Sterling\lmcsorcid{0000-0002-0585-5564}}
\author[D.~Gratzer]{Daniel Gratzer\lmcsorcid{0000-0003-1944-0789}}
\author[L.~Birkedal]{Lars Birkedal\lmcsorcid{0000-0003-1320-0098}}

\address{Aarhus University}
\email{jsterling@cs.au.dk, gratzer@cs.au.dk, birkedal@cs.au.dk}

\begin{abstract}
  We contribute the first denotational semantics of polymorphic dependent type
  theory extended by an equational theory for general (higher-order)
  reference types and recursive types, based on a combination of guarded
  recursion and impredicative polymorphism; because our model is based on
  \emph{recursively defined semantic worlds}, it is compatible with
  polymorphism and relational reasoning about stateful abstract datatypes. We
  then extend our language with modal constructs for proof-relevant relational
  reasoning based on the \emph{logical relations as types} principle, in which
  equivalences between imperative abstract datatypes can be established
  synthetically. What is new in relation to prior typed denotational models of
  higher-order store is that our Kripke worlds need not be syntactically
  definable, and are thus compatible with relational reasoning in the heap. Our
  work combines recent advances in the operational semantics of state with the
  purely denotational viewpoint of \emph{synthetic guarded domain theory}.
\end{abstract}

\maketitle

\section{Introduction}

The combination of parametric polymorphism and general reference types in
denotational semantics is notoriously difficult, though neither feature
presents any serious difficulties on its own.

\begin{enumerate}

  \item Parametric polymorphism can be modelled in a \emph{complete internal
    category} \`a la Hyland~\cite{hyland:1988}; for instance, the category of partial
    equivalence relations has ``large'' products over the category of
    assemblies, which contains the assembly of all partial equivalence
    relations.

  \item General reference types can be modeled via a Kripke world-indexed state
    monad, where the worlds assign \emph{syntactic} types to locations in the
    heap, as in the possible worlds model of thunk storage by
    Levy~\cite{levy:2003:book,levy:2002}.

\end{enumerate}

The two techniques described above are not easily combined: in order to support
parametric reasoning in the presence of reference types, it is necessary at a
minimum for the Kripke worlds to assign \emph{semantic} types to locations
rather than only syntactic types. But a semantic type should itself be a family
of (sets, predomains, \etc) indexed in Kripke worlds; thus one attempts to
solve a domain equation of that defines a preorder $\WICat$ of Kripke worlds simultaneously with
the category of functors~\cite{reynolds:1981,oles:1986} from $\WICat$ to the category of predomains:
\[
  \WICat \cong {\Con{Loc} \finto \Con{SemType}}
  \qquad
  \Con{SemType} \cong \brk{\WICat,\Con{Predomain}}
  \tag{\textcolor{red}{\textasteriskcentered}}
\]

There are already a few problems with the ``domain equation'' presented above.
For one, the collection of predomains is not itself a predomain, but it could
be replaced by the domain of finitary projections of some universal
domain~\cite{coquand-gunter-winskel:1994}. The more fundamental problem is
that it remains unclear how to interpret type connectives on $\Con{SemType}$;
for instance, Birkedal, St\o{}vring, and
Thamsborg~\cite[\S5]{birkedal-stovring-thamsborg:2010} provide an explicit
counterexample to demonstrate that a na\"ive interpretation of the reference
type cannot coexist with recursive types in the presence of semantic worlds.

\subsection{Step-indexing and guarded domain theory}\label{sec:guarded-domain-theory}

The difficulties outlined above have been side-stepped by the introduction of
\emph{step-indexing}~\cite{ahmed:2004,appel-mellies-richards-vouillon:2007}
and its semantic counterpart \emph{metric/guarded domain
theory}~\cite{america-rutten:1987,rutten-turi:1992,brssty:2011}, which proceed
by stratifying recursive definitions in their finite approximations. This
stratification was axiomatized by Birkedal~\etal~\cite{bmss:2011} in \DefEmph{synthetic
guarded domain theory} / SGDT via the \DefEmph{later modality}
$\Mor[\Ltr]{\mathscr{S}}{\mathscr{S}}$ which comes equipped with a point
$\Mor[\Next]{\ArrId{\mathscr{S}}}{\Ltr}$; the standard model of SGDT is the topos
of trees $\mathscr{S} = \Psh{\omega}$ where the later modality is defined like
so:
\[
  \prn{\Ltr{A}}_n = \Lim{k < n} A_k
  \qquad
  \Next[A]^nx = k\mapsto  x\Sub{\vert k}
\]

In the setting of SGDT, any recursive definition has a unique fixed
point if its recursive variables are guarded by an application of $\Ltr$. In
categorical terms, $\mathscr{S}$ is \emph{algebraically compact} with respect
to \emph{locally contractive} endofunctors. Letting $\TpSet$ be a type
universe in a model $\mathscr{S}$ of synthetic guarded domain theory, it is
easy to solve the following approximate domain equation:
\[
  \mathbb{W} \cong \Con{Loc}\finto {\color{RegalBlue}\Ltr}\Con{SemType}\qquad
  \Con{SemType} \cong \brk{\mathbb{W}, \TpSet}
\]

The preorder $\mathbb{W}$ of Kripke worlds and the collection of semantic types
$\Con{SemType}$ both exist as objects of $\mathscr{S}$; it is even possible to
define a suitable connective $\Mor[\Con{ref}]{\Con{SemType}}{\Con{SemType}}$.
We may also define an indexed type of heaps $\mathbb{H} : \mathbb{W}\to
\TpSet$; unfortunately it is quite unclear to define the
\emph{indexed state monad} for $\mathbb{H}$ as a connective on $\Con{SemType}$.
For instance, the standard definition of the indexed state monad~\cite{plotkin-power:2002,levy:2004,power:2011}
runs into obvious size problems given that
$\Con{SemType}\not\in \TpSet$ and thus $\WICat\not\in\TpSet$:
\[
  \prn{\TpT{A}}w =
  \Prod{\color{red} w'\geq w}
  \mathbb{H}w' \to
  \Sum{\color{red} w''\geq w'}
  \mathbb{H}w''\times Aw''
  \tag{\textcolor{red}{\textasteriskcentered}}
\]

Indeed, the definition above cannot be executed in the standard model
$\mathscr{S} = \Psh{\omega}$ when $\TpSet$ is the Hofmann--Streicher lifting of
a Grothendieck universe from $\SET$. This problem is side-stepped in
\emph{operational} models, where $\TpSet$ is replaced by a set of predicates on
syntactical expressions of the programming language being modeled and $\TpT*$
is defined by a guarded version of weakest preconditions as in the original
paper of Birkedal \etal~\cite{bmss:2011}:
\begin{align*}
  \TpT{A}w &= \Compr{u \in \Con{Val}}{
    \forall w' \geq w,
    h\in \mathbb{H}w'.
    \Con{wp}\prn{w',h,u\TmUnit} \brc{\prn{w'',h',v}\mapsto  v\in Aw''}
  }
  \\
  \Con{wp}\prn{w,h,e} \brc{\Phi} &=
  \prn{e\in\Con{Val}\land \Phi\prn{w,h,e}}
  \lor
  \exists w', h', e'.
  \prn{h;e}\mapsto \prn{h';e'}
  \land
  \Ltr \Con{wp}\prn{w',h'} \brc{\Phi}
\end{align*}

Operational methods worked in $\mathscr{S}$ \emph{only because} the set of
predicates is a complete lattice: we may compute the join and intersection of
arbitrarily large families of predicates. This observation will ultimately form
the basis for our own non-operational solution; rather than giving up
on denotations, we will work in a \emph{different} model of synthetic
guarded domain theory that contains an \DefEmph{impredicative universe}
$\TpSet$, \ie one closed under universal and (thus) existential types \`a la
System~F; then we can define the indexed state monad as follows
without encountering size problems:
\[
  \prn{\TpT{A}}w =
  {\color{RegalBlue}\Forall{w'\geq w}} \mathbb{H}w' \to
  {\color{RegalBlue}\Exists{w''\geq w'}} \mathbb{H}w''\times Aw''
\]

This approach raises the question: \emph{does there in fact exist a model of
synthetic guarded domain theory with an impredicative universe?} We answer in
the affirmative, constructing a model in presheaves on a well-founded order
internal to a \DefEmph{realizability topos}~\cite{van-oosten:2008}. Our model
of SGDT is thus an instance of the relative-topos-theoretic generalization of
synthetic guarded domain theory introduced by Palombi and Sterling~\cite{palombi-sterling:2023}.

\subsection{This paper: impredicative guarded dependent type theory with reference types}

Our denotational model of higher-order store was immediately suited for
generalization to dependent types; this scalability is one of many advantages
our abstract category-theoretic methods.
Justified by this model, we have defined Impredicative Guarded Dependent Type
Theory (\iGDTT{}) and its extension with general reference types (\iGDTTRef{})
and their equational theory.

\iGDTTRef{} is the first language to soundly combine ``full-spectrum''
dependently typed programming and equational reasoning with both higher-order
store and recursive types. The necessity of \emph{semantics} for higher-order
store with dependent types is not hypothetical: both Idris~2 and Lean~4 are
dependently typed, higher-order functional programming
languages~\cite{brady:2021,de-moura-ullrich:2021} that feature a Haskell-style
\texttt{IO} monad with general \texttt{IORef} types~\cite{jones:2001}, but
until now these features have had no semantics.  In fact, our investigations
have revealed a potential problem in the \texttt{IO} monads of both Idris and
Lean: as higher-order store seems to be inherently impredicative, it may not
make sense to close multiple nested universes under the \texttt{IO}
monad~\cite{coquand:1986}.

Many real-world examples require one to go \emph{beyond} simple
equational reasoning; to this end, we defined an extension \iGDTTRefLRAT{}
with constructs for synthetic, proof-relevant relational reasoning for data
abstraction and weak bisimulation based on the \DefEmph{logical relations as
types} (LRAT) principle of Sterling and Harper~\cite{sterling-harper:2021}, which axiomatizes a
generalization of the parametricity translation of dependent type
theory~\cite{bernardy-jansson-paterson:2012,bernardy-moulin:2012}.

\subsection{Discussion of related work}

\NewDocumentCommand\LNLD{}{\textbf{LNL}\textsubscript{\textsc{d}}}
\NewDocumentCommand\iHTT{}{\textbf{iHTT}}

\subsubsection{Operational semantics of higher-order store}

The most thoroughly developed methods for giving semantics to higher-order
store are based on \emph{operational semantics}~\cite{ahmed:2004};
accompanying the operational semantics of higher-order store is a wealth of
powerful program logics for modular reasoning about higher-order effectful and
even concurrent programs based on higher-order separation
logic~\cite{svendsen-birkedal:2014,iris:2018,bizjak-birkedal:2018}.

We are motivated to pursue denotations for three reasons. First, denotational
methods are amenable to the use of general theorems from other mathematical
fields to solve difficult domain problems, whereas operational methods tend to
require one to solve every problem ``by hand'' without relying on standard
lemmas. Secondly, a good denotational semantics tends to simplify reasoning
about programs, as the exponents of the \textsc{DeepSpec} project have argued
\cite{itrees:popl:2019}. Finally and most profoundly, there is an emerging
need to program directly with \emph{actual spaces}, as in differentiable and
probabilistic programming
languages~\cite{abadi-plotkin:2019,vakar-kammar-staton:2019}.

Although operational and denotational semantics are different in both purpose
and technique, there is nonetheless a rich interplay between the two
traditions. First of all, the idea of step-indexing and its approximate
equational theory lies at the heart of our denotational semantics; second, the
wealth of results in operationally based models and program logics for
higher-order store suggest several areas for future work in our denotational
semantics. For instance, it would be interesting to rebase a program logic such
as Iris atop our synthetic denotational model; likewise, several perationally
based
works~\cite{ahmed-dreyer-rossberg:2009,dreyer-neis-rossberg-birkedal:2010}
suggest many improvements to our notion of \emph{semantic world} to support
richer kinds of correspondence between programs.

\subsubsection{Denotational semantics of state}

There is a rich tradition of denotational semantics of state, ranging from a
single reference cell~\cite{moggi:1991} to first-order
store~\cite{plotkin-power:2002} and storage of
pointers~\cite{kammar-levy-moss-staton:2017}, and even higher-order store with
syntactic worlds~\cite{levy:2004}. None of these approaches is compatible
with relational reasoning for polymorphic/abstract types, as they all
rely on the semantic heap being classified by syntactically definable types.

\paragraph{Untyped metric semantics and approximate locations}

A somewhat different \emph{untyped} approach to the denotational semantics of
higher-order store with recursive types and polymorphism was pioneered by
Birkedal, St\o{}vring, and Thamsborg~\cite{birkedal-stovring-thamsborg:2010} in
which one uses metric/guarded domain theory to define a \emph{universal domain}
$\Con{D}$, and then develops a model of System~\FMuRef{} in predicates on
$\Con{D}$. Because predicates form a complete lattice, it is possible to
interpret the state operations as we have discussed in
\cref{sec:guarded-domain-theory}. An important aspect of this class of models
is the presence of \emph{approximate locations}, which \opcit{} have shown to
be non-optional.\footnote{Indeed, even in our own model the presence of
approximate locations can be detected using Thamsborg's observation on the
correspondence between metric and ordinary domain
theory~\cite[Ch.~9]{thamsborg:2010}.}

\paragraph{Comparison}

Our own model resembles a synthetic version of that of Birkedal, St\o{}vring,
and Thamsborg~\cite{birkedal-stovring-thamsborg:2010}, but there are some
important differences. Our model is considerably more abstract and more general
than that of \opcit; whereas the cited work must solve a very complex metric
domain equation to construct a universal domain, we may use the realizability
topos generated by \emph{any} partial combinatory algebra, and thus we do not
depend on any specific choice of universal domain.  Moreover, we argue that
ours is a \emph{typed} model: we depend only on the combination of guarded
recursion and an impredicative universe, and although the principle source of
such structures is realizability on untyped partial combinatory
algebras~\cite{lietz-streicher:2002}, we do not depend on any of the details of
realizability. Because of the generality and abstractness of our model
construction, scaling up to full dependent type theory has offered no
resistance.

\subsubsection{Higher-order store in Impredicative Hoare Type Theory}

Realizability models of impredicative type theory have been used before to give
a model of so-called Impredicative Hoare Type Theory (\iHTT), an extension of
dependent type theory with a monadic ``Hoare'' type for stateful computations
with higher-order store~\cite{svendsen-birkedal-nanevski:2011}.  However, in
contrast to our monadic type in \iGDTTRef, \iHTT{} does not support equational
reasoning about computations: all elements of a Hoare type are equated and one
can only reason via the types. Moreover, the Hoare type in \opcit only supports
untyped locations with substructural reference capabilities that change as
the heap evolves (the so-called ``strong update''); thus non-Hoare types are not indexed in worlds.

\subsubsection{Linear dependent type theory}

Another approach to the integration of state into dependent type theory is
contributed by Krishnaswami, Pradic, and
Benton~\cite{krishnaswami-pradic-benton:2015}, who develop an adjoint
linear--non-linear dependent type theory \LNLD{} with a realizability model in
partial equivalence relations. Like \iHTT, the theory of \opcit uses
capabilities for references with strong update; unlike \iHTT{}, the account of
store in \LNLD{} enjoys a rich equational theory. What is missing from \LNLD{}
is any account of general recursion, which normally arises from higher-order
store via backpatching or Landin's Knot; indeed, \LNLD{} carefully avoids the
general recursive aspects of higher-order store via linearity. One of the main
advances of our own paper over both \iHTT{} and \LNLD{} is to model a
dependently typed equational theory for \emph{full} store with general
recursion.

\subsubsection{Effectful dependent type theory}

We have not attempted any non-trivial interaction between computational
effects and the dependent type structure of our language, in contrast to the
work of P\'edrot and Tabareau~\cite{pedrot-tabareau:2020} on dependent call-by-push-value: our
\iGDTTRef{} language can be thought of as a purely functional programming
language extended with a monad for stateful programming with Haskell-style
\texttt{IORef}s.

\subsection{Structure and contributions of this paper}

Our contributions are as follows:

\begin{itemize}

  \item In \textbf{\cref{sec:igdtt}} we introduce \DefEmph{impredicative
    guarded dependent type theory} (\iGDTT) as a user-friendly metalanguage for
    the denotational semantics of languages involving polymorphism, general
    reference types and recursive types.  In \textbf{\cref{sec:easy}} as a case
    study, we construct a simple denotational model of Monadic~System~\FMuRef,
    a language with general reference types, polymorphic types, and recursive
    types in \iGDTT{}. Finally in \textbf{\cref{sec:coq}} we describe our Coq
    library for \iGDTT{} in which we have formalized the higher-order state monad and
    reference types.

  \item In \textbf{\cref{sec:dependent}} we describe \Alert{\iGDTTRef{}}, an
    extension of \iGDTT{} with general reference types and a monad for
    higher-order store. \iGDTTRef{} is thus a full-spectrum dependently typed
    programming language with support for higher-order effectful programming.
    To illustrate the combination of higher-order store with dependent types,
    we define and prove the correctness of an implementation of factorial
    defined via \DefEmph{Landin's knot} / backpatching.

  \item In \textbf{\cref{sec:igdtt-lrat}} we describe \Alert{\iGDTTRefLRAT{}}, an
    extension of \iGDTTRef{} with constructs for synthetic proof-relevant
    relational reasoning based on the \DefEmph{logical relations as types}
    principle~\cite{sterling-harper:2021}. \iGDTTRefLRAT{} can be used to
    succinctly exhibit bisimulations between higher-order stateful computations
    and abstract data types, which we demonstrate in two case studies involving
    imperative counter implementations
    (\cref{sec:case-study:local-references,sec:case-study:adt}).

  \item In \textbf{\cref{sec:semantics}}, we describe general results for
    constructing models of \iGDTT{} and \iGDTTRefLRAT, with concrete
    instantiations given by a combination of realizability and internal
    presheaves. The results of this section justify the consistency of
    \iGDTTRefLRAT{} as a language for relational reasoning about higher-order
    stateful computations.

  \item In \textbf{\cref{sec:concl}} we conclude with some reflections on directions for future work.

\end{itemize}

\subsection{Acknowledgments}

We are thankful to Frederik Lerbjerg Aagaard, Robert Harper, Rasmus
M\o{}gelberg for helpful conversations concerning this project.
This work was supported in part by a Villum Investigator grant (no. 25804),
Center for Basic Research in Program Verification (CPV), from the VILLUM
Foundation.
Jonathan Sterling is funded by the European Union under the Marie
Sk\l{}odowska-Curie Actions Postdoctoral Fellowship project
\href{https://cordis.europa.eu/project/id/101065303}{\emph{TypeSynth: synthetic
methods in program verification}}. Views and opinions expressed are however
those of the authors only and do not necessarily reflect those of the European
Union or the European Commission. Neither the European Union nor the granting
authority can be held responsible for them.
\NewDocumentCommand\StateEquations{}{
  \ebrule{
    \hypo{l:\TpRef{A}}
    \hypo{u:A}
    \infer2{
      \TmSet{A}{l}{u};
      \TmGet{A}{l}
      =
      \Alert{\TmStep};
      \TmSet{A}{l}{u};
      \TmRet{u}
      :
      \TpT{A}
    }
  }
  \and
  \ebrule{
    \hypo{l:\TpRef{A}}
    \infer1{
      \prn{\TmBind{x}{\TmGet{A}{l}} \TmSet{A}{l}{x}}
      =
      \Alert{\TmStep}
      :
      \TpT\TpUnit
    }
  }
  \and
  \ebrule{
    \hypo{u,v:A}
    \infer1{
      \prn{
        \TmBind{x}{\TmNew{A}{u}}
        \TmSet{A}{x}{v};
        \TmRet{x}
      }
      =
      \TmNew{A}{v}
      : \TpT\prn{\TpRef{A}}
    }
  }
  \and
  \ebrule{
    \hypo{l:\TpRef{A}}
    \hypo{u,v:A}
    \infer2{
      \TmSet{A}{l}{u};\TmSet{A}{l}{v}
      =
      \TmSet{A}{l}{v}
      :\TpT\TpUnit
    }
  }
}

\section{Impredicative guarded dependent type theory}\label{sec:igdtt}

In this section we describe an extension of \emph{guarded dependent type
theory} with an impredicative universe; guarded dependent type theory is a
dependently typed interface to synthetic guarded domain theory. The purpose of
this \DefEmph{impredicative guarded dependent type theory} (\iGDTT{}) is to
serve as a metalanguage for the denotational semantics of programming languages
involving general reference types, just as ordinary guarded dependent type
theory can be used as a metalanguage for denotational semantics of programming
languages with recursive functions~\cite{paviotti-mogelberg-birkedal:2015} and
recursive types~\cite{mogelberg-paviotti:2016}.

\subsection{Universe structure: quantifiers and reflection}

The core \iGDTT{} language is modelled off of Martin-L\"of type theory with a
pair of impredicative base universes $\prn{\TpProp\subseteq\TpSet}\in
\TpType_0\in\TpType_1\in\ldots$, as in the version of the calculus of inductive
constructions with \emph{impredicative Set}.\footnote{We mean that every
element of $\TpProp$ is also classified by $\TpSet$, but we do not assert that
$\TpProp$ is classified by $\TpSet$.} The universes $\TpSet,\TpType_i$ are
closed under dependent products, dependent sums, finite enumerations $\brk{n}$,
inductive types (W-types), and extensional equality types with equality
reflection.
We assert that $\TpProp$ is both proof-irrelevant and univalent, and moreover
closed under extensional equality types.  Proof-irrelevance means that for any
$P:\TpProp$ and $p,q:P$ we have $p=q$; univalence means that if
$P\leftrightarrow Q$ then $P=Q$.\footnote{This is also called
\emph{propositional extensionality}.}

What makes $\TpSet,\TpProp$ impredicative is that we
assert an additional connective for \emph{universal types} with abstraction,
application, $\beta$-, and $\eta$-laws.
\begin{mathpar}
  \ebrule[impredicativity]{
    \hypo{\mathbb{S}\in \brc{\TpSet,\TpProp}}
    \hypo{A : {\color{RegalBlue}\TpType_i}}
    \hypo{x:A\vdash Bx : \mathbb{S}}
    \infer3{\Forall{x : A}{Bx} : \mathbb{S}}
  }
\end{mathpar}

The universal type automatically gives rise to an impredicative encoding of
\emph{existential} types. The na\"ive encoding $\Exists{x:A}{Bx}
\eqdef\Sup{\color{red}*} \Forall{C : \mathbb{S}}\prn{\Forall{x:A}\prn{Bx\to C}}\to C$
does not in fact have the correct universal property as its $\eta$-law holds
only up to parametricity, but we may use the method of
Awodey, Frey, and Speight~\cite{awodey-frey-speight:2018} to define a correct version of the existential
type. It will be simplest to do so in two steps: first define the
\emph{reflection} $\color{RegalBlue}\vvrt{-}\Sub{\mathbb{S}} : \TpType_i\to \mathbb{S}$, and then
apply this reflection to the dependent sum.

\begin{thm}
  The inclusion $\EmbMor{\mathbb{S}}{\TpType_i}$ has a left adjoint
  $\vvrt{-}\Sub{\mathbb{S}} : \TpType_i \to \mathbb{S}$.
\end{thm}
\begin{proof}
  For reasons of space, we give only the definition of
  $\vvrt{-}\Sub{\mathbb{S}}$. The reflection is constructed in two steps; first
  we define the ``wild'' reflection
  $\vrt{-}\Sub{\mathbb{S}} : \TpType_i\to\mathbb{S}$ by an impredicative
  encoding, which is unfortunately too unconstrained to have the universal
  property of the left adjoint:

  \iblock{
    \mrow{
      \vrt{-}\Sub{\mathbb{S}} : \TpType_i\to\mathbb{S}
    }
    \mrow{
      \vrt{A}\Sub{\mathbb{S}} \eqdef \Forall{C:\mathbb{S}}\prn{A\to C}\to C
    }
  }

  We therefore constrain $\vrt{A}\Sub{\mathbb{S}}$ by a naturality condition, encoded as a
  structure $\Con{ok}_A : \vrt{A}\Sub{\mathbb{S}}\to \mathbb{S}$ defined using universal and equality
  types like so:

  \iblock{
    \mrow{
      \Con{ok}_A : \vrt{A}\Sub{\mathbb{S}}\to\mathbb{S}
    }
    \mrow{
      \Con{ok}_A \alpha \eqdef
      \Forall{C,D:\mathbb{S}}
      \Forall{f : C \to D}
      \Forall{h : A \to C}
      \alpha\,D\, \prn{f \circ h} = f \prn{\alpha\, C\,h}
    }
    \row
    \mrow{
      \vvrt{-}\Sub{\mathbb{S}} : \TpType_i\to\mathbb{S}
    }
    \mrow{
      \vvrt{A}\Sub{\mathbb{S}} \eqdef \Sum{\alpha : \vrt{A}\Sub{\mathbb{S}}}{\Con{ok}_A\alpha}
    }
  }

  This completes the construction of the reflection.
\end{proof}

We will write $\Con{pack} : A \to \vvrt{A}\Sub{\mathbb{S}}$ for the unit of the reflection; via
the universal property of the adjunction, maps into types classified by $\mathbb{S}$ can be defined
by pattern matching, \eg $\Kwd{let}\,\Con{pack}\,x = u\,\Kwd{in}\ v$; by virtue
of the constraint $\Con{ok}_A$, these destructuring expressions satisfy a desirable
$\eta$-law.
With the reflection in hand, it is possible to give a correct encoding of the
existential type:
\begin{mathpar}
  \ebrule{
    \hypo{A : {\color{RegalBlue}\TpType_i}}
    \hypo{x:A\vdash Bx : \mathbb{S}}
    \infer2{
      \Exists{x:A}{Bx} \eqdef
      \vvrt{\Sum{x:A}{Bx}}\Sub{\mathbb{S}}
    }
  }
\end{mathpar}

The na\"ive impredicative encoding would have worked for $\TpProp$ because it
is already proof-irrelevant; but for uniformity, we present the reflections for
$\TpProp,\TpSet$ simultaneously.

\subsection{The later modality and delayed substitutions}

The remainder of \iGDTT{}---the guarded fragment---is the same in prior
presentations~\cite{bgcmb:2016}; to summarize, we have a type operator $\Ltr$
called the \DefEmph{later modality} equipped with a fixed-point operator. We
defer a proper exposition of guarded type theory to
Bizjak~\etal~\cite{bgcmb:2016}; intuitively, however, the type $\Ltr A$
contains elements of $A$ which only become available one `step' in the future.
We have chosen to present $\Ltr$ using the \DefEmph{delayed substitutions}
\fbox{$\xi\leadsto\Xi$} of \opcit:
\begin{mathpar}
  \ebrule{
    \hypo{\xi \leadsto \Xi}
    \hypo{
      \Xi \vdash A\ \mathit{type}
    }
    \infer2{
      \Ltr\brk{\xi}.A\ \mathit{type}
    }
  }
  \and
  \ebrule{
    \hypo{\xi\leadsto\Xi}
    \hypo{\Xi\vdash a:A}
    \infer2{\Next*\brk{\xi}.{a} : \Ltr\brk{\xi}.{A}}
  }
  \and
  \ebrule{
    \hypo{\vphantom{\Xi}}
    \infer1{\cdot\leadsto\cdot}
  }
  \and
  \ebrule{
    \hypo{\xi \leadsto \Xi}
    \hypo{a : \Ltr\brk{\xi}.A}
    \infer2{\prn{\xi, x\leftarrow a} \leadsto \Xi, x:A}
  }
\end{mathpar}

The delayed substitution attached to the introduction and formation rules
allows an element of $\Ltr A$ to strip away the $\Ltr$-modalities from a list
of terms, thereby rendering $\Ltr$ an applicative functor in the sense of
McBride and Paterson~\cite{mcbride-paterson:2008}:
\[
  f \circledast a \eqdef \Next*\brk{x \leftarrow f, y \leftarrow a}.x{y}
\]

All universes $\mathbb{X}\in\brc{\TpSet,\TpProp,\TpType_i}$ are closed under $\Ltr$.
We have equational laws governing both delayed substitutions and the
$\Ltr/\Next*$ constructors; we will not belabor them here, referring instead to
\opcit for a precise presentation. In the case of empty delayed substitutions
$\xi = \cdot$, we will write $\Ltr{A}$ and $\Next{a}$ for $\Ltr\brk{\cdot}.A$
and $\Next*\brk{\cdot}.a$ respectively. A guarded fixed point combinator
is included:
\begin{mathpar}
  \ebrule{
    \hypo{x : \Ltr{A} \vdash fx : A}
    \infer1{\Con{gfix}\, x. fx : A}
  }
  \and
  \ebrule{
    \hypo{x : \Ltr{A} \vdash fx : A}
    \infer1{\Con{gfix}\, x. fx  = f\,\prn{\Next\prn{\Con{gfix}\,x.fx}}: A}
  }
\end{mathpar}

\subsection{Guarded domains and the lift monad}\label{sec:guarded-domains}

By construction, all types in \iGDTT{} support a guarded fixed-point operator
$\prn{\Ltr A \to A} \to A$; without an algebra structure $\Ltr A \to A$,
however, such a fixed point operator is insufficient to interpret general
recursion. We will define a \DefEmph{guarded domain} to be a type equipped with
exactly such an algebra structure below.

\begin{defi}
  We define a \DefEmph{guarded domain} to be a type $A$ together with a
  function $\vartheta_A : \Ltr{A}\to A$, \ie an algebra for the later modality
  viewed as an endofunctor.
\end{defi}

Given a guarded domain $A$, we define the \DefEmph{delay map} $\delta_A : A \to
A$ to be the composite $\vartheta_A\circ\Next*$.  A guarded domain can be equipped with
a fixed point combinator $\mu_A : \prn{A\to A}\to A$ satisfying
$\mu_Af = \delta_A\,\prn{f\,\prn{\mu_A f}}$. In particular, we define
$\mu_A f \eqdef \Con{gfix}\, x.\, \vartheta_A\,\prn{\Next*[A]\brk{z\leftarrow x}.{fz}}$.

\begin{exa}
  Each universe $\mathbb{X}\in\brc{\TpProp,\TpSet,\TpType_i}$ carries the structure of
  a guarded domain, as we may define
  $\vartheta\Sub{\mathbb{X}} A \eqdef \Ltr\brk{X\leftarrow A}.X$.
  Note that we have $\delta\Sub{\mathbb{X}}A = \Ltr{A}$.
\end{exa}

Because the universe is itself a guarded domain, the same fixed point
combinators can be used to interpret \emph{both} recursive programs \emph{and}
recursive types as pointed out by Birkedal and
M\o{}gelberg~\cite{birkedal-mogelberg:2013}; this is a significant improvement
over ordinary (synthetic) domain theory, where more complex notions of
algebraic compactness are required to lift recursion to the level of types.

\NewDocumentCommand\GDom{m}{\Kwd{GDom}\prn{#1}}

\begin{con}[Guarded lift monad]
  We will write $\GDom{\mathbb{X}}$ for the universe of guarded domains in a universe
  $\mathbb{X}$; the forgetful functor
  $\Mor[\Con{J}]{\GDom{\mathbb{X}}}{\mathbb{X}}$ has a left adjoint
  ${\color{RegalBlue}\TpL*} \dashv \Con{J}$ that freely \emph{lifts} a type to
  a guarded domain, which we may compute by taking a guarded fixed point:
  \[
    \TpL{A} = A + \Ltr\TpL{A}
    \qquad\qquad
    \vartheta\Sub{\TpL{A}} \eqdef \Con{inr}
  \]

  Unfolding definitions, we have solved the guarded domain equation $\TpL{A} =
  A + \Ltr\TpL{A}$; we will write $\eta : A \to \TpL{A}$ for the left
  injection.
  When it causes no confusion, we will leave the forgetful functor $\Con{J}$
  implicit; thus we refer to $\Mor[\TpL*]{\mathbb{X}}{\mathbb{X}}$ as the
  \DefEmph{guarded lift monad}.
\end{con}

\subsection{Case study: denotational semantics of Monadic System~\texorpdfstring{\FMuRef}{F/mu/ref}}\label{sec:easy}

Already we have enough machinery to explore a simple possible worlds model of
System~\FMuRef, a polymorphic language with general reference types. Even in this
simple case, the tension between polymorphism and general reference types is
evident and we require both the impredicative and guarded-recursive features of
\iGDTT{} to construct the \DefEmph{semantic worlds} that underly the model.
We emphasize that the entire construction takes place internally to \iGDTT{}.

For the moment, we only sketch the interpretation of types in our model. We will
return to this point in \cref{sec:semantics}, when we use a similar definition
of semantic worlds to construct a model for a version of full \iGDTT{} extended
with general references.

\subsubsection{Recursively defined semantic worlds and heaps}\label{sec:worlds-and-heaps}

We define a function $\Mor[\Con{World}]{\TpType_0}{\TpType_0}$ taking a type
$A:\TpType_0$ to the preorder of finite maps $\Nat\finto A$, where the order is
given by graph inclusion.
Next we define $\mathcal{T}:\TpType_0$ to be the guarded fixed point of the
operator that sends $A:\Ltr\TpType_0$ to the type of functors
$\brk{\Con{World}\prn{\Ltr\brk{z\leftarrow A}z},\TpSet}$; thus we have
$\Alert{\mathcal{T}=\brk{\Con{World}\prn{\Ltr{\mathcal{T}}},\TpSet}}$. Finally we
define $\WICat$ to be the preorder $\Con{World}\prn{\Ltr{\mathcal{T}}}$. Thus
we have solved the domain equation $\Alert{\WICat=\Nat\finto\Ltr\brk{\WICat,\TpSet}}$ and we have $\Alert{\mathcal{T}=\Ob{\brk{\WICat,\TpSet}}}$.
The object $\mathcal{T}$ can be seen to be a guarded domain, setting
$
  \vartheta\Sub{\mathcal{T}}A \eqdef \lambda w. \Ltr\brk{X\leftarrow A}.Xw
$.

We extend the above to a functor
$\Mor[\Con{heap}]{\OpCat{\WICat}}{\brk{\WICat,\TpSet}}$ as follows, writing
$\OpCat{\WICat}$ for the opposite of the preorder $\WICat$, by defining
$
  \Con{heap}\,w w' \eqdef
  \Prod{i\in \vrt{w}}
  \vartheta\Sub{\mathcal{T}}\, \prn{wi}\,w'
$.
Here we have used $\vrt{w}$ to denote the support of the finite mapping $w$.
We will write $\HICat_w$ for $\Con{heap}\,w\,w$ and $\HICat:\TpType_0$ for the
dependent sum $\Sum{w:\Ob{\WICat}}\HICat_w$.

\subsubsection{Higher-order state monad and reference types}\label{sec:monad}

We now define a strong monad $\TpT*$ on $\brk{\WICat,\TpSet}$, a variation on
the standard state monad suitable for computations involving state and general
recursion. Recalling that $\TpL*:\TpSet\to\TpSet$ is the guarded lift monad, we
define the higher-order state monad $\TpT*$ below:

\iblock{
  \mrow{
    \TpT*:\brk{\WICat,\TpSet}\to\brk{\WICat,\TpSet}
  }
  \mrow{
    \prn{\TpT{A}}w \eqdef
    \Forall{w'\geq w}
    {
      \HICat\Sub{w'}\to
      \TpL*
      \Exists{w''\geq w'}
      \HICat\Sub{w''}\times Aw''
    }
  }
}

\begin{thm}\label{thm:t-monad}
  $\TpT*$ is a strong monad and each $\prn{\TpT{A}}w$ is a guarded domain.
\end{thm}

Next we define a semantic reference type connective
$\TpRef* : \mathcal{T} \to \mathcal{T}$. For each $A : \mathcal{T}$, the type
$\TpRef{A}$ picks out those locations in the current world which
contain elements of type $A$:

\iblock{
  \mrow{
    \TpRef*:\mathcal{T}\to\mathcal{T}
  }
  \mrow{
    \prn{\TpRef{A}}w \eqdef
    \Compr{l\in\vrt{w}\DelimMin{1}}{
      \Next{A} = wl
    }
  }
}

\subsubsection{Synthetic model of Monadic~System~\texorpdfstring{\FMuRef}{F/mu/ref}}

We may now implement a denotational semantics of Monadic~System~\FMuRef{} in
\iGDTT{}; the direct-style call-by-value version of System~\FMuRef can then be
interpreted separately \`a la Moggi in the standard way. We specify the domains
of interpretation for each form of judgment below:
\begingroup\small
\begin{gather*}
  \bbrk{\JdgKCtx{\Xi}} : \TpType\quad
  \bbrk{\JdgTCtx{\Xi}{\Gamma}},\bbrk{\JdgTp{\Xi}{\tau}} : \bbrk{\Xi}\to \mathcal{T}
  \quad
  \bbrk{\JdgTm{\Xi}{\Gamma}{e:\tau}} : \Forall{\xi:\bbrk{\JdgKCtx{\Xi}}} \bbrk{\JdgTCtx{\Xi}{\Gamma}}\xi \to \bbrk{\JdgTp{\Xi}{\tau}}\xi
\end{gather*}
\endgroup

The structure of kind and type contexts is interpreted below using the
cartesian structure of the functor categories $\brk{\WICat,\TpType}$ and
$\brk{\WICat,\TpSet}$ respectively:

\iblock{
  \mrow{
    \bbrk{\JdgKCtx{\cdot}} \eqdef \TpUnit
  }
  \mrow{
    \bbrk{\JdgKCtx{\Xi,\alpha}} \eqdef \bbrk{\JdgKCtx{\Xi}}\times \mathcal{T}
  }
  \row
  \mrow{
    \bbrk{\JdgTCtx{\Xi}{\cdot}}\xi \eqdef \TpUnit
  }
  \mrow{
    \bbrk{\JdgTCtx{\Xi}{\Gamma,x:\tau}}\xi \eqdef \bbrk{\JdgTCtx{\Xi}{\Gamma}}\xi \times \bbrk{\JdgTp{\Xi}{\tau}}\xi
  }
}

Individual type connectives are interpreted below:

\iblock{
  \mrow{
    \bbrk{\JdgTp{\Xi}{\forall \alpha. \tau}}\xi w \eqdef
    \Forall{X:\mathcal{T}}
    \bbrk{\JdgTp{\Xi}{\tau}}\prn{\xi,X} w
  }
  \mrow{
    \bbrk{\JdgTp{\Xi}{\exists \alpha. \tau}}\xi w \eqdef
    \Exists{X:\mathcal{T}}
    \bbrk{\JdgTp{\Xi}{\tau}}\prn{\xi,X} w
  }
  \mrow{
    \bbrk{\JdgTp{\Xi}{\exists \alpha. \tau}}\xi w \eqdef
    \Exists{X:\mathcal{T}}
    \bbrk{\JdgTp{\Xi}{\tau}}\prn{\xi,X} w
  }
  \mrow{
    \bbrk{\JdgTp{\Xi}{\mu \alpha. \tau}} \xi
    \eqdef
    \mu\Sub{\mathcal{T}}\,\prn{\lambda X. \bbrk{\JdgTp{\Xi,\alpha}{\tau}}\prn{\xi,X}}
  }
  \mrow{
    \bbrk{\JdgTp{\Xi}{\TpRef{\tau}}}\xi \eqdef
    \TpRef\prn{\bbrk{\JdgTp{\Xi}{\tau}}\xi}
  }
  \mrow{
    \bbrk{\JdgTp{\Xi}{\TpT{\tau}}}\xi \eqdef
    \TpT\prn{\bbrk{\JdgTp{\Xi}{\tau}}\xi}
  }
}

Function and product types are interpreted using the cartesian closure of
$\brk{\WICat,\TpSet}$; note that functor categories with relatively large
domain are not typically cartesian closed, but in our case the impredicativity
of $\TpSet\in\TpType$ ensures cartesian closure.
Next we interpret the state operations as families of natural transformations
in $\brk{\WICat,\TpSet}$:

\iblock{
  \mhang{
    \bbrk{\JdgTm{\Xi}{\Gamma}{\TmGet{\tau}{l}:\TpT{\tau}}}\xi w\gamma w' h\eqdef
  }{
    \mrow{
      \Kwd{let}\ i \eqdef \bbrk{\JdgTm{\Xi}{\Gamma}{l:\TpRef{\tau}}}\xi w' \gamma\Sub{\vert w'};
    }
    \mrow{
      \vartheta\,\prn{
        \Next*\brk{z\leftarrow h i}.\,\eta\,\prn{\Con{pack}\,\prn{w',\prn{h,z}}}
      }
    }
  }

  \row

  \mhang{
    \bbrk{\JdgTm{\Xi}{\Gamma}{\TmSet{\tau}{l}{e}: \TpT{\TpUnit}}}\xi w \gamma w' h \eqdef
  }{
    \mrow{
      \Kwd{let}\ i \eqdef \bbrk{\JdgTm{\Xi}{\Gamma}{l:\TpRef{\tau}}}\xi w' \gamma\Sub{\vert w'};
    }
    \mrow{
      \Kwd{let}\ x \eqdef \bbrk{\JdgTm{\Xi}{\Gamma}{e:\tau}}\xi w' \gamma\Sub{\vert w'}
    }
    \mrow{
      \eta\,\prn{
        \Con{pack}\,\prn{
          w',
          \prn{
            h\brk{
              i \mapsto \Next {x}
            },
            \TmUnit
          }
        }
      }
    }
  }

  \row

  \mhang{
    \bbrk{\JdgTm{\Xi}{\Gamma}{\TmNew{\tau}{e}:\TpT\prn{\TpRef{\tau}}}}\xi w \gamma w' h \eqdef
  }{
    \mrow{
      \Kwd{let}\ i = \Con{fresh}\,\vrt{w'};
    }
    \mrow{
      \Kwd{let}\ w'' \eqdef w' \cup\brc{i\mapsto \Next\prn{\bbrk{\JdgTp{\Xi}{\tau}}}};
    }
    \mrow{
      \Kwd{let}\ x \eqdef \bbrk{\JdgTm{\Xi}{\Gamma}{e:\tau}}\xi w'' \gamma\Sub{\vert w''};
    }
    \mrow{
      \eta\,\prn{
        \Con{pack}\,\prn{
          w'',
          \prn{
            h\Sub{\vert{}w''}\cup\brc{l\mapsto\Next{x}},
            i
          }
        }
      }
    }
  }
}

In the interpretation of the $\TmNew{\tau}$ operator, we have assumed a
deterministic ``allocator'' $\Con{fresh} : \Ob{\WICat} \to \mathbb{N}$ that
chooses an unused address in a given world. Note that the naturality conditions
of the interpretation do not depend in any way on the behavior of
$\Con{fresh}$.

\subsubsection{Abstract steps in the equational theory of higher-order store}

Notice that our interpretation of the getter must invoke the $\Ltr$-algebra
structure on $\TpT{A}$ as the heap associates to each location $i$ an
element of $\vartheta\Sub{\mathcal{T}}\prn{w i}$. For this reason, equations
governing \emph{reads} to the heap do not hold on the nose but rather hold up
to an abstract ``step''.

For example, given $l : \TpRef{A}$ we might expect the equation
$\bbrk{\TmSet{A}{l}{a}; \TmGet{A}{l}} = \bbrk{\TmSet{A}{l}{a}; \TmRet{a}}$ to hold, but
the call to $\vartheta$ in the former denotation will prevent them from
agreeing in our model. This behavior is familiar already from the
work of Escard\'o~\cite{escardo:1999} and Paviotti~\etal~\cite{paviotti-mogelberg-birkedal:2015} on denotational semantics of
general recursion in guarded/metric domain theory: the guarded interpretation of recursion forces a more
intensional notion of equality that counts abstract steps. Thus in order to
properly formulate the equations that \emph{do} hold, we must account for these
abstract steps. To this end, we extend Monadic~System~\FMuRef{} by a new
primitive effect $\Alert{\TmStep:\TpT{\TpUnit}}$ that takes an abstract step:

\iblock{
  \mrow{
    \bbrk{\JdgTm{\Xi}{\Gamma}{\TmStep:\TpT{\TpUnit}}}\xi w \gamma w' h \eqdef
    \Alert{\delta}\,\prn{
      \eta\,\prn{
        \Con{pack}\,\prn{w',\prn{h,\TmUnit}}
      }
    }
  }
}

With the above in hand, we may state and prove the expected equations for state:
\begin{thm}
  Our denotational semantics of Monadic~System~\FMuRef validates the following equations:
  \begin{mathpar}
    \ebrule{
      \hypo{\JdgTm{\Xi}{\Gamma}{l:\TpRef{\tau}}}
      \hypo{\JdgTm{\Xi}{\Gamma}{u:\tau}}
      \infer2{
        \JdgTm{\Xi}{\Gamma}{
          \TmSet{\tau}{l}{u};
          \TmGet{\tau}{l}
          \equiv
          \Alert{\TmStep};
          \TmSet{\tau}{l}{u};
          \TmRet{u}
          :
          \TpT{\tau}
        }
      }
    }
    \and
    \ebrule{
      \hypo{\JdgTm{\Xi}{\Gamma}{l:\TpRef{\tau}}}
      \infer1{
        \JdgTm{\Xi}{\Gamma}{
          \prn{
            \TmBind{x}{\TmGet{\tau}{l}}
            \TmSet{\tau}{l}{x}
          }
          \equiv \Alert{\TmStep}
          : \TpT\TpUnit
        }
      }
    }
    \and
    \ebrule{
      \hypo{\JdgTm{\Xi}{\Gamma}{u,v:\tau}}
      \infer1{
        \JdgTm{\Xi}{\Gamma}{
          \prn{
            \TmBind{x}{\TmNew{\tau}{u}}
            \TmSet{\tau}{x}{v};
            \TmRet{x}
          }
          \equiv
          \TmNew{A}{v}
          : \TpT\prn{\TpRef{\tau}}
        }
      }
    }
    \and
    \ebrule{
      \hypo{\JdgTm{\Xi}{\Gamma}{l:\TpRef{\tau}}}
      \hypo{\JdgTm{\Xi}{\Gamma}{u,v:\tau}}
      \infer2{
        \JdgTm{\Xi}{\Gamma}{
          \TmSet{\tau}{l}{u};\TmSet{\tau}{l}{v}
          \equiv
          \TmSet{\tau}{l}{v}
          :\TpT\TpUnit
        }
      }
    }
  \end{mathpar}
\end{thm}

\begin{proof}
  Immediate upon unfolding the denotational semantics of the given equations.
\end{proof}

\subsection{Formal case study: presheaf model of higher-order store in Coq}\label{sec:coq}

Although the semantics of \iGDTT{} are admittedly quite sophisticated (see
\cref{sec:semantics}), one of its advantages is that it is quite
straightforward to extend existing proof assistants such as Coq~\cite{coq:reference-manual} and Agda~\cite{norell:2009} with
the axioms of \iGDTT{}. Thus a practitioner of programming language semantics
can define ``na\"ive'' synthetic models of effects such as higher-order store
in a proof assistant without needing to know the category theory that was
required to invent and justify \iGDTT{}.

To validate the use of \iGDTT{} for formal semantics of programming languages,
we have postulated the axioms of \iGDTT{} in a Coq library and used it to
formalize the construction of the indexed state monad $\TpT*$ defined in
\cref{sec:monad} as well as the reference type connective. In our
formalization, we decompose $\Con{T}$ as the horizontal composition of several
adjunctions and thus obtain the monad laws for free by taking $\Con{T}$ to be
the monad of the composite adjunction.

\subsubsection{Axiomatizing the impredicative universe}

We take Coq's \cd{Set} universe to be the impredicative universe $\TpSet$ of \iGDTT{};
it is possible to execute Coq with the flag \verb|-impredicative-set|, but we
have found it simpler to implement the impredicativity by means of a local pragma
as follows:
\begin{code}
  \#[bypass\_check(universes = yes)]
  \CoqKwd{Definition} All \{A : Type\} \{B : A -> Set\} : Set :=
    forall x : A, B x.
\end{code}

For \iGDTT's predicative universes $\TpType_i$ we use Coq's predicative
universes \cd{Type@\{i\}}. We formalize the left adjoint to the inclusion
$\cd{Set}\hookrightarrow \cd{Type@\{i\}}$, verifying its universal property
using the argument of Awodey, Frey, and Speight~\cite{awodey-frey-speight:2018}.

\subsubsection{Axiomatizing guarded recursion}

Next we axiomatize the later modality in Coq. Coq's implicit universe
polymorphism ensures that the following axioms land not only in \cd{Type}
but in \cd{Set} as well:
\begin{code}
  \CoqKwd{Axiom} later : Type -> Type.
  \CoqKwd{Notation} "$\Ltr$ A" := (later A) (at level 60).
  \CoqKwd{Axiom} next : forall {A}, A -> $\Ltr$ A.
  \CoqKwd{Axiom} loeb : forall {A} (f : $\Ltr$ A -> A), A.
  \CoqKwd{Axiom} loeb\_unfold : forall {A} (f : $\Ltr$ A -> A), loeb f = f (next (loeb f)).
\end{code}

We do not include the general delayed substitutions in our axiomatization,
because they are too difficult to formalize; for our case study, the following
special case of delayed substitutions suffices:
\begin{code}
  \CoqKwd{Axiom} dlater : $\Ltr$ Type -> Type.
  \CoqKwd{Axiom} dlater\_next\_eq : forall A, $\Ltr$ A = dlater (next A).
\end{code}

\subsubsection{Formalizing guarded interaction trees}

Building on the above and several other axioms of guarded recursion that we
omit for lack of room, we formalize a notion of guarded algebraic effect based
on \emph{containers}~\cite{abbott-altenkirch-ghani:2005}. Given a container
$\mathbb{E} = \prn{S : \TpSet, P : S \to \TpSet}$, we may view each $s:S$ as
the name of an effect operation for which $Ps$ is the continuation arity.
Writing $\mathcal{P}\Sub{\mathbb{E}}X = \Sum{s:S}\Prod{p:Ps}X$ for the
polynomial endofunctor presented by $\mathbb{E}$, we may consider for each $A$
the free $\mathcal{P}\Sub{\mathbb{E}}$ algebra generated by $A$ and the
(co)-free $\mathcal{P}\Sub{\mathbb{E}}$ coalgebra generated by $A$; the former
is the inductive type of finite (terminating) computations of a value of type
$A$ that use the effects specified by $\mathbb{E}$, whereas the latter is the
\emph{coinductive} type of possibly infinite computations. The latter are
referred to as \emph{interaction trees} by Xia~\etal~\cite{itrees:popl:2019},
who have propelled a resurgence in their use for reasoning about code in
first-order languages via denotational semantics.

We formalize a \emph{guarded} version of interaction trees, obtained by solving
the guarded domain equation $X = A + \mathcal{P}\Sub{\mathbb{E}}\Ltr{X}$ using \cd{loeb} in \cd{Set}.

\begin{code}
  \CoqKwd{Record} Container := \{op : Set; bdry : op -> Set\}.

  \CoqKwd{Inductive} ITree\_F ($\mathbb{E}$ : Container) A T : Set) : Set :=
  | Ret (r : A)
  | Do (e : $\mathbb{E}$) (k : bdry e -> T).

  \CoqKwd{Definition} ITree ($\mathbb{E}$ : Container) (A : Set) : Set :=
    loeb (fun T => ITree\_F (dlater T)).
\end{code}

We develop a free-forgetful adjunction for each container $\mathbb{E}$; in
particular, we prove that \cd{ITree $\mathbb{E}$ A} is the free
$\prn{\mathcal{P}\Sub{\mathbb{E}}\circ\Ltr}$-algebra generated by $A$. The
reason for developing guarded interaction trees is that there is a particular
container $\mathbb{E}$ for which \cd{ITree $\mathbb{E}$} is the \emph{guarded
lift monad} $\TpL*$ of Paviotti, M\o{}gelberg, and
Birkedal~\cite{paviotti-mogelberg-birkedal:2015}, which we have discussed in
\cref{sec:guarded-domains}. Nonetheless, the extra generality allows us to
modularly add other effects such as failure or printing, \etc.

\subsubsection{Recursively defined worlds, presheaves, and the state monad}

We start by defining the category of finite maps of elements of a given type:
\begin{code}
  \CoqKwd{Definition} world : Type -> Category.type := ...
\end{code}

Then we define the collections of semantic worlds and semantic types by solving a
guarded domain equation, where we write \cd{SET} for the category induced by
the impredicative \cd{Set} universe.
\begin{code}
  \CoqKwd{Definition} $\mathcal{F}$ (T : Type) : Type := [world T, SET].
  \CoqKwd{Definition} $\mathcal{T}$ : Type := loeb (fun T => $\mathcal{F}$ (dlater T)).
  \CoqKwd{Definition} $\mathbb{W}$ : Category.type := world $\mathcal{T}$.
  \CoqKwd{Definition} heap (w : $\mathbb{W}$) : Type := ...
  \CoqKwd{Definition} $\mathbb{H}$ : Type := \{w : $\mathbb{W}$ & heap w\}.
\end{code}

We have also formalized the reference type connective as a map \cd{ref :
$\mathcal{T}$ -> $\mathcal{T}$}, exactly as in the informal presentation of
\cref{sec:monad}.  Our formalization of the state monad differs from our
informal presentation in two ways: firstly we have generalized over an
arbitrary container $\mathbb{E}$ specifying computational effects (\eg
printing, failure, \etc), and secondly we found it simpler to decompose it into
the following sequence of simpler adjunctions, whose definition involves the
impredicative \cd{All} quantifier that we axiomatized earlier:
\[
  \begin{tikzpicture}[diagram]
    \node (0) {$\brk{\OpCat{\WICat},\cd{Alg $\mathbb{E}$}}$};
    \node (1) [right = 2.8cm of 0] {$\brk{\OpCat{\WICat},\cd{Set}}$};
    \node (2) [right = 2.5cm of 1] {$\brk{\HICat,\cd{Set}}$};
    \node (3) [right = 2.5cm of 2] {$\brk{\WICat,\cd{Set}}$};

    \draw[->,bend right=30] (0) to node (01-s) [below] {$\brk{\OpCat{\WICat},\cd{Forget}}$} (1);
    \draw[->,bend right=30] (1) to node (01-n) [above] {$\brk{\OpCat{\WICat},\cd{ITree $\mathbb{E}$}}$} (0);
    \node [between=01-s and 01-n] {$\bot$};

    \draw[->,bend right=30] (1) to node (12-s) [below] {$\Reix{\cd{heap}}$} (2);
    \draw[->,bend right=30] (2) to node (12-n) [above] {$\Exists{\cd{heap}}$} (1);
    \node [between=12-s and 12-n] {$\bot$};

    \draw[->,bend right=30] (2) to node (23-s) [below] {$\Forall{\cd{heap}}$} (3);
    \draw[->,bend right=30] (3) to node (23-n) [above] {$\Reix{\cd{heap}}$} (2);
    \node [between=23-s and 23-n] {$\bot$};

  \end{tikzpicture}
\]

We did not formalize the state operations,
but there is no obstacle to doing so.
\section{Impredicative guarded dependent type theory with reference types}\label{sec:dependent}

In this section, we detail an extension \Alert{\iGDTTRef{}} of \iGDTT{} with
general reference types and a state monad. \iGDTTRef{} can serve as both an
effectful higher-order programming language in its own right, \emph{and} as a
dependently typed metalanguage for the denotational semantics of other
programming languages involving higher-order store. Later on we will extend
\iGDTTRef{} with relational constructs to enable the \emph{verification} of
stateful programs inside the type theory. Denotational semantics for these
extensions of \iGDTT{} are discussed in \cref{sec:semantics}.

\subsection{Adding reference types to \texorpdfstring{\iGDTT{}}{iGDTT}}\label{sec:dependent:axiomatization}

In this section, we extend the \iGDTT{} language with the following constructs:
\begin{enumerate}
  \item a monad $\TpT* : \TpSet\to\TpSet$,
  \item a later algebra structure $\vartheta\Sub{\TpT{A}}:\Ltr{\TpT{A}}\to\TpT{A}$ parameterized in $A:\TpSet$ such that:
    \begin{mathpar}
      \ebrule{
        \hypo{u:\TpT{A}}
        \hypo{x:A\vdash vx:\TpT{B}}
        \infer2{
          \delta\Sub{\TpT{B}}\prn{\TmBind{x}{u} vx}
          =
          \prn{\TmBind{x}{\delta\Sub{\TpT{A}}u} vx}
          =
          \prn{\TmBind{x}{u} \delta\Sub{\TpT{B}}\prn{vx}}
          :
          \TpT{B}
        }
      }
    \end{mathpar}

    We will write $\TmStep : \TpT{\TpUnit}$ for the generic effect $\delta\Sub{\TpT{\TpUnit}}\prn{\TmRet{\TmUnit}}$; the rule above ensures that $\TmStep$ commutes with all operations in the monad.

  \item a reference type connective closed under the following rules:
    \begin{mathpar}
      \ebrule{
        \hypo{A:\TpSet}
        \infer1{\TpRef{A}:\TpSet}
      }
      \and
      \ebrule{
        \hypo{l:\TpRef{A}}
        \infer1{
          \TmGet{A}{l}
          : \TpT{A}
        }
      }
      \and
      \ebrule{
        \hypo{l:\TpRef{A}}
        \hypo{u:A}
        \infer2{
          \TmSet{A}{l}{u}:\TpT{\TpUnit}
        }
      }
      \and
      \ebrule{
        \hypo{u:A}
        \infer1{
          \TmNew{A}{u} : \TpT\prn{\TpRef{A}}
        }
      }
      \and
      \StateEquations
    \end{mathpar}
\end{enumerate}

\subsection{Programming with higher-order store and dependent types}

\subsubsection{Type dependency in the heap}

Because the universe $\TpSet$ is closed under dependent types (including
equality types, \etc), there is no obstacle to storing elements of dependent
types in the heap. For instance, it is easy to allocate a reference that
can only hold even integers; as an example, we consider a program that increments an
even integer in place, abbreviating $T\eqdef \Compr{z:\TpInt}{\Con{isEven}\,z}$:

\iblock{
  \mrow{
    M : \TpRef{T}\to \TpT{\TpUnit}
  }
  \mrow{
    M \eqdef
    \lambda l.
    \TmBind{x}{\TmGet{T}{l}}
    \TmSet{T}{l}{\prn{x+2}}
  }
}

\NewDocumentCommand\TpNat{}{\mathbb{N}}

\subsubsection{Recursion via back-patching}

To illustrate higher-order store, we can give a more sophisticated example
involving \emph{backpatching} via Landin's knot, writing $\bot :
\TpL{\alpha}$ for the divergent element $\mu x. x$.

\iblock{
  \mrow{
    \Con{patch} : \Forall{\alpha:\TpSet} \prn{\prn{\alpha\to \TpT{\alpha}} \to \prn{\alpha\to\TpT{\alpha}}} \to \TpT\prn{\TpRef\prn{\alpha\to \TpT{\alpha}}}
  }
  \mrow{
    \Con{patch}\,\alpha\,F \eqdef
    \TmBind{r}{\TmNew{\alpha\to \TpT{\alpha}}{\prn{\lambda \_. \bot}}}
    \TmSet{\alpha\to\TpT{\alpha}}{r}{
      \prn{F\,\prn{\lambda x. \TmBind{f}{\TmGet{\alpha\to\TpT{\alpha}}{r}} fx}}
    };
    \TmRet{r}
  }

  \row

  \mrow{
    \Con{knot} : \Forall{\alpha:\TpSet} \prn{\prn{\alpha\to \TpT{\alpha}} \to \prn{\alpha\to\TpT{\alpha}}} \to {\alpha\to \TpT{\alpha}}
  }

  \mrow{
    \Con{knot}\,\alpha\,F\,x \eqdef
    \TmBind{r}{\Con{patch}\,\alpha\,F}
    F\,\prn{\lambda z. \TmBind{f}{\TmGet{\alpha\to\TpT{\alpha}}{r}} fz}\, x
  }
}

We can use the backpatching knot to give a (contrived) computation of the factorial:

\iblock{
  \mrow{
    \Con{fact}' : \prn{\TpNat\to\TpT{\TpNat}} \to \TpNat\to \TpT{\TpNat}
  }
  \mrow{
    \Con{fact}'\,f\,0 \eqdef \TmRet{1}
  }
  \mrow{
    \Con{fact}'\,f\,\prn{n+1} \eqdef
    \TmBind{m}{f\, n}
    \TmRet\prn{\prn{n+1} \times m}
  }
  \row
  \mrow{
    \Con{fact} : \TpNat\to\TpT{\TpNat}
  }
  \mrow{
    \Con{fact} \eqdef \Con{knot}\,\TpNat\,\Con{fact}'
  }
}

\begin{lem}
  Entirely inside of \iGDTTRef{} we may prove the
  following correctness lemma given a reference implementation $\Con{goodFact} :
  \TpNat\to\TpNat$:
  \[
    \Forall{n:\TpNat}
    \Con{fact}\,n = \prn{
      \TmBind{\_}{\Con{patch}\,\TpNat\,\Con{fact}'} \TmStep^n; \TmRet\prn{\Con{goodFact}\,n}
    }
  \]
\end{lem}

\begin{proof}
  We proceed by induction on $n$; the base case is immediate:
  \begin{align*}
    \Con{fact}\,0 &=
    {r \leftarrow \Con{patch}\, \TpNat\,\Con{fact}'; \Con{fact}'\,\prn{\ldots}\, 0}
    \\
    &= {\_ \leftarrow \Con{patch}\, \TpNat\,\Con{fact}'; \TmRet{1}}
    \\
    &= {\_ \leftarrow \Con{patch}\, \TpNat\,\Con{fact}'; \TmRet\prn{\Con{goodFact}\,0}}
  \end{align*}

  Next we fix $n:\mathbb{N}$ such that $\Con{fact}\,n = \prn{\TmBind{\_}{\Con{patch}\,\TpNat\,\Con{fact}'} \TmStep^n; \TmRet\prn{\Con{goodFact}\,n}}$:
  \begin{align*}
    &\Con{fact}\,\prn{n+1}
    \\
    &\quad=
    {
      \TmBind{r}{\Con{patch}\, \TpNat\,\Con{fact}'}
      \Con{fact}'\,\prn{\lambda x. \TmBind{f}{\TmGet{\TpNat\to\TpT{\TpNat}}{r}} fx}\, \prn{n+1}
    }
    \\
    &\quad=
    {
      \TmBind{r}{\Con{patch}\, \TpNat\,\Con{fact}'}
      \TmBind{m}{
        \prn{
          \TmBind{f}{\TmGet{\TpNat\to\TpT{\TpNat}}{r}} {fn}
        }
      }
      \TmRet\prn{\prn{n+1} \times m}
    }
    \\
    &\quad=
    {
      \TmBind{r}{\Con{patch}\, \TpNat\,\Con{fact}'}
      \TmBind{f}{\TmGet{\TpNat\to\TpT{\TpNat}}{r}}
      \TmBind{m}{fn}
      \TmRet\prn{\prn{n+1} \times m}
    }
    \\
    &\quad=
    {
      \TmBind{r}{\Con{patch}\, \TpNat\,\Con{fact}'}
      \TmStep;
      \TmBind{m}{\Con{fact}'\,\prn{\lambda x. \TmBind{f}{\TmGet{\TpNat\to\TpT{\TpNat}}{r}}{fx}}\,n}
      \TmRet\prn{\prn{n+1} \times m}
    }
    \\
    &\quad=
    {
      \TmBind{r}{\Con{patch}\, \TpNat\,\Con{fact}'}
      \TmBind{m}{\Con{fact}'\,\prn{\lambda x. \TmBind{f}{\TmGet{\TpNat\to\TpT{\TpNat}}{r}}{fx}}\,n}
      \TmStep;
      \TmRet\prn{\prn{n+1} \times m}
    }
    \\
    &\quad=
    m\leftarrow \Con{fact}\,n;
    \TmStep;
    \TmRet\prn{\prn{n+1}\times m}
    &&\text{by def.}
    \\
    &\quad=
    \TmBind{\_}{\Con{patch}\,\TpNat\,\Con{fact}'} \TmStep^n;
    \TmStep;
    \TmRet\prn{\prn{\Con{goodFact}\,n + 1} \times n}
    &&
    \text{by i.h.}
    \\
    &\quad=
    \TmBind{\_}{\Con{patch}\,\TpNat\,\Con{fact}'} \TmStep\Sup{n+1};
    \TmRet\prn{\Con{goodFact}\,\prn{n + 1}}
    &&\qedhere
  \end{align*}
\end{proof}

\begin{rem}
  Note that we do not have
  $\Con{fact}\,n=\TmRet\prn{\Con{goodFact}\,n}$; to understand why this is the
  case, there are two things to take note of. First, the abstract steps
  taken by each access to the heap are explicitly tracked by the equational
  theory of \iGDTTRef{} as an instance of the $\TmStep$ effect. Secondly, even
  though the allocations carried out by $\Con{patch}$ are no longer active upon
  return, the equational theory of \iGDTTRef{} cannot ``garbage collect'' them.
\end{rem}

\section{Logical relations as types for recursion and higher-order store}\label{sec:igdtt-lrat}

Like any dependent type theory, \iGDTTRef{} is equipped with a ``built-in''
notion of equality, but it is too fine to serve as a basis for reasoning about
the equivalence of effectful programs. For instance, we have already seen that
even the most basic programs governing reading and writing to a reference may
differ in the number of steps they take. Accordingly, we will construct a coarser
notion of equivalence between programs in the form of a logical relation over
\iGDTTRef{}.

In particular, we discuss an extension of the axioms of \iGDTTRef{} which
supports proof-relevant binary relational reasoning in the sense of
Sterling and Harpers's \DefEmph{logical relations as types} (LRAT)
principle~\cite{sterling-harper:2021}.\footnote{The LRAT language is also referred to in other contexts as
\emph{synthetic Tait computability}~\cite{sterling:2021:thesis}.} The
result is \iGDTTRefLRAT{}, a new dependent type theory with lightweight
features for verifying the equivalence of effectful computations. In addition
to capturing the traditional ``method of candidates'' reasoning,
\iGDTTRefLRAT{} includes a notion of weak
bisimulation~\cite{mogelberg-paviotti:2016,paviotti:2016}. The inclusion of
weak bisimulation is crucial; without it, only programs that read from the heap
in lockstep could ever be related.

Finally, we show how our synthetic approach simplifies two standard examples
from Birkedal, St\o{}vring, and Thamsborg~\cite{birkedal-stovring-thamsborg:2010} concerning equivalent
implementations of imperative counters.

\subsection{Logical relations as types: a synthetic approach to relational reasoning}

The \DefEmph{logical relations as types} principle states that logical
relations can be treated \emph{synthetically} as ordinary types in the presence
of certain modal operators that exist in categories of logical relations. Thus
LRAT is an abstraction of logical relations that avoids the bureaucracy of
analytic approaches; the ``work'' that goes into conventional logical relations
arguments has not disappeared, but has been refactored into general results of
category theory that are much simpler to establish directly than almost any of
their practical consequences in programming languages.

We execute the LRAT extension \Alert{\iGDTTRefLRAT} of \iGDTTRef{} by
postulating three propositions $\OpnL,\OpnR,\Opn:\TpProp$ representing the
left- and right-hand sides of a correspondence, and the disjoint union of both
sides; the idea is that these propositions will generate modal operators that
allow us to think of \emph{any} type as a correspondence, and then project out
different aspects of the correspondence.

\begin{mathpar}
  \ebrule{
    \infer0{\OpnL,\OpnR,\Opn : \TpProp}
  }
  \and
  \ebrule{
    \hypo{\_:\OpnL}
    \hypo{\_:\OpnR}
    \infer2{\_:\bot}
  }
\end{mathpar}

Above we have imposed a rule to make $\OpnL,\OpnR$ disjoint. We additionally
add rules to make $\Opn$ behave as the \emph{coproduct} $\OpnL+\OpnR$, with a
splitting rule not only for propositions but also for types.

\begin{mathpar}
  \ebrule{
    \hypo{\_:\OpnL\vdash a_L:A}
    \hypo{\_:\OpnR\vdash a_R:A}
    \infer2{
      \_:\Opn\vdash \brk{\OpnL\hookrightarrow a_L,\OpnR\hookrightarrow a_R} : A
    }
  }
  \and
  \ebrule{
    \hypo{\_:\OpnL\vdash a_L:A}
    \hypo{\_:\OpnR\vdash a_R:A}
    \infer2{
      \_:\OpnL\vdash \brk{\OpnL\hookrightarrow a_L,\OpnR\hookrightarrow a_R} = a_L : A
    }
  }
  \and
  \ebrule{
    \hypo{\_:\OpnL\vdash a = b:A}
    \hypo{\_:\OpnR\vdash a = b :A}
    \infer2{
      \_:\Opn\vdash a = b : A
    }
  }
\end{mathpar}

\begin{rem}
  In other applications of LRAT such as those of
  Sterling~\etal~\cite{sterling-angiuli:2021,sterling-harper:2022}, case splits
  were included for arbitrary $\phi\lor\psi$. This is possible in the internal language of a topos, but
  \emph{not} in the \iGDTT{} where the universe of propositions is meant to lie
  within $\TpType_i$, which will be interpreted by a quasitopos rather than a
  true topos.  Roughly the issue is that in a quasitopos $\ECat$, the union of
  two regular subobjects need not be regular unless $\ECat$ is
  quasiadhesive~\cite{johnstone-lack-sobocinski:2007}.
\end{rem}

The modalities that we consider in this paper are all \emph{left exact,
idempotent, monadic} modalities in the sense of
Rijke, Shulman, and Spitters~\cite{rijke-shulman-spitters:2020}. Left exactness means that they preserve
finite limits, and idempotence means that the multiplication map $T^2\to T$ is
a natural isomorphism. In this section, let $T$ be an arbitrary left exact
idempotent monad.

\begin{defi}[Modal types]
  We will refer to a type $A$ as
  \DefEmph{$T$-modal} when the unit map $\Mor[\eta_A]{A}{TA}$ is an isomorphism.
\end{defi}

\begin{nota}[Modal elimination]
  Idempotent monadic modalities enjoy a universal property: for any $T$-modal
  type $C$, the function $\Mor{\prn{T A \to C}}{\prn{A\to C}}$ induced by
  precomposition with the unit is an isomorphism.  We reflect this in our
  notation as follows:
  \begin{mathpar}
    \ebrule{
      \hypo{C\text{ is $T$-modal}}
      \hypo{u : TA}
      \hypo{x : A \vdash fx : C}
      \infer3{
        \prn{\TmBind{x}{u} fx} : C
      }
    }
    \and
    \ebrule{
      \hypo{\ldots}
      \hypo{a:A}
      \hypo{x:A\vdash fx:C}
      \infer3{
        \prn{\TmBind{x}{\eta a} fx} = fa: C
      }
    }
    \and
    \ebrule{
      \hypo{C\text{ is $T$-modal}}
      \hypo{x:TA\vdash fx, gx : C}
      \hypo{x:A\vdash f\,\prn{\eta x} = g\,\prn{\eta x} : C}
      \hypo{u:TA}
      \infer4{
        fu = gu :C
      }
    }
  \end{mathpar}
\end{nota}

Let $\phi:\TpProp$ be a proposition.

\begin{defi}
  A type $A$ is called \DefEmph{$\phi$-transparent} when the constant map
  $\Mor{A}{\prn{\phi\to A}}$ is an isomorphism; on the other hand, $A$ is
  called \DefEmph{$\phi$-sealed} when the projection map $\Mor{\phi\times
  A}{\phi}$ is an isomorphism. Equivalently, $A$ is $\phi$-sealed exactly when
  $\phi$ implies that $A$ is a singleton; we will write $\star : A$ for the
  unique element of any $\top$-sealed type $A$.
\end{defi}

Intuitively, a $\phi$-transparent type is one that ``thinks'' $\phi$ is true;
conversely, a $\phi$-sealed type is one that ``thinks'' $\phi$ is false and
thus contracts to a point under $\phi$.
The $\phi$-transparent and $\phi$-sealed types are each governed by modalities,
referred to as \DefEmph{open} and \DefEmph{closed} respectively.

\begin{defi}[Open modality]
  We will write $\OpMod{\phi}{A}$ for the \emph{implicit} function space
  $\phi\to A$, which we refer to as the \DefEmph{open modality} for $\phi$; we
  will leave both the $\lambda$-abstraction and application implicit in our
  notation. Likewise, given an element $A:\OpMod{\phi}{\TpSet}$ we will write
  $\OpMod{\phi}{A}$ for the \emph{implicit} dependent product of $A$; this is
  the \DefEmph{dependent open modality}.\footnote{The implicitness is only a
  matter of convenient notation; it plays no mathematical role.}
\end{defi}

\begin{defi}[Closed modality]\label{def:closed-modality}
  We will write $\ClMod{\phi}{A}$ for the following quotient inductive type,
  which we shall call the \DefEmph{closed modality} associated to $\phi$:

  \iblock{
    \mhang{
      \Kwd{quotient}\ \Kwd{data}\ \ClMod{\phi}{A} : \TpSet \ \Kwd{where}
    }{
      \mrow{
        \eta\Sub{\ClMod{\phi}} : A \to \ClMod{\phi}{A}
      }
      \mrow{
        \star : \brc{\_:\phi}\to {\ClMod{\phi}{A}}
      }
      \mrow{
        \_ : \brc{\_:\phi,u:\ClMod{\phi}{A}} \to u = \star
      }
    }
  }

  (We use curly braces to indicate \emph{implicit} arguments.)
  Put another way, $\ClMod{\phi}{A}$ is the quotient of the coproduct $\phi+A$
  under the $\TpProp$-valued equivalence relation $u\sim v \Longleftrightarrow \phi\lor\prn{u=v}$.
  The elimination rule for $\ClMod{\phi}{A}$ is thus a case statement with a side condition:
  \begin{mathpar}
    \ebrule{
      \hypo{x:\ClMod{\phi}{A}\vdash Cx : \TpSet}
      \hypo{x:A\vdash fx : C\,\prn{\eta x}}
      \hypo{\_:\phi\vdash g : C\star}
      \hypo{\Alert{x:A,\_:\phi\vdash fx = g : C\star}}
      \hypo{u : \ClMod{\phi}{A}}
      \infer5{
        \Kwd{case}\ u\ \Kwd{of}\ \eta\Sub{\ClMod{\phi}}x \hookrightarrow fx \mid \star \hookrightarrow g : Cu
      }
    }
  \end{mathpar}

  Note that we have only assumed an elimination rule for motives valued in $\TpSet$.
\end{defi}

\begin{rem}[Effectivity]\label{rem:effectivity}
  Any equivalence relation valued in $\TpProp$ is effective, as $\TpProp$
  satisfies propositional extensionality; thus the quotient in
  \cref{def:closed-modality} is effective. The only subtlety worth noting
  is that although the equivalence relation $\phi\lor\prn{u=v}$ is by
  definition a pushout in $\TpProp$ of the projections $\phi\leftarrow
  \phi\land\prn{u=v}\rightarrow u=v$, this pushout \emph{need not} be preserved by
  the inclusion $\EmbMor{\TpProp}{\TpSet}$; an example of this behavior in semantics is that
  the union of two regular subobjects need not be regular.
\end{rem}

\begin{fact}\label{fact:tr-sl-modalities}
  A type $A:\TpSet$ is $\phi$-transparent if and only if it is modal for the open
  modality $\OpMod{\phi}{-}$; the type $A:\TpSet$ is $\phi$-sealed if and only if it
  is modal for the closed modality $\ClMod{\phi}{-}$.
\end{fact}

In light of \cref{fact:tr-sl-modalities}, it is instructive to point out that the
elimination form for the closed modality
  $\prn{\TmBind{x}{u} fx}$ can be implemented by
$
  \Kwd{case}\ u\ \Kwd{of}\ \eta\Sub{\ClMod{\phi}}x \hookrightarrow fx \mid \star\hookrightarrow \star
$.

\paragraph{Extension types}
Given a type $A$, a proposition $\phi$ and an element $a:\OpMod{\phi}{A}$ we
will write $\Compr{A}{\phi\hookrightarrow a}$ for the subtype
$\Compr{x:A}{\OpMod{\phi}{x=a}}\subseteq A$. We will write
$\Compr{A}{\phi\hookrightarrow a,\psi\hookrightarrow b,\ldots}$ for the
extension type $\Compr{A}{\phi\lor \psi\lor \ldots\hookrightarrow
\brk{\phi\hookrightarrow a,\psi\hookrightarrow b,\ldots}}$ when $\phi,\psi$ are disjoint.

\begin{rem}[LRAT axiomatizes the parametricity translation]
  In the synthetic setting, to each type $A$ we may associate the span
  $\prn{\OpMod{\OpnL}{A}}\longleftarrow {A}
  \longrightarrow\prn{\OpMod{\OpnR}{A}}$ given by the units of the open
  modalities for $\OpnL,\OpnR$; this should be thought of as the span that $A$ is
  taken to in the binary parametricity translation. Indeed, given $u\SubL
  :\OpMod{\OpnL}{A}$ and $u\SubR:\OpMod{\OpnR}{A}$ the subtype $\Compr{A}{\OpnL\hookrightarrow u\SubL, \OpnR\hookrightarrow u\SubR}$ can be
  thought of as the type of proofs that $u\SubL$ and $u\SubR$ are related in the
  correspondence.
  It is in this sense that LRAT axiomatizes the parametricity translation.
\end{rem}

A final axiom of the synthetic relational
extension of \iGDTT{} is the \emph{refinement type}.%
\footnote{%
  In prior presentations of synthetic logical relations~\cite{sterling-harper:2021,sterling-angiuli:2021},
  the role played here by the refinement type connective (first axiomatized by Sterling and Harper~\cite{sterling-harper:2022}) was instead played by the so-called
  \emph{realignment} axiom. The two are
  interderivable~\cite[\S5]{gratzer-shulman-sterling:2022:universes}, but the refinement type connective
  is both more direct and provides stronger geometrical intuition.
}
Given a type
$A:\OpMod{\Opn}{\TpSet}$ and a family of $\Opn$-sealed types
$B:\prn{\OpMod{\Opn}{A}}\to\TpSet$, we may consider
$\Sum{x:\OpMod{\Opn}{A}}B x$. The \DefEmph{refinement type} is a new code for
this dependent sum, written $\brk{\Opn\hookrightarrow x:A \mid Bx}$, that
enjoys the following \emph{additional} principle of type equality: the
refinement type becomes equal to $A$ in the presence of $\_:\Opn$.
In other words, the refinement type is governed by the interface presented in \cref{fig:refinement-type}.
We have not assumed that any proposition $\phi:\TpProp$ besides $\Opn$
is equipped with a refinement type connective.

\begin{figure}
  \begin{mathpar}
    \ebrule[formation]{
      \hypo{A:\OpMod{\Opn}{A}}
      \hypo{B:\prn{\OpMod{\Opn}{A}}\to \TpSet}
      \hypo{\OpMod{\Opn}{\forall x:A.\ {Bx \cong 1}}}
      \infer3{
        \brk{\Opn\hookrightarrow x:A\mid Bx} : \Compr{\TpSet}{\Opn\hookrightarrow A}
      }
    }
    \and
    \ebrule[introduction]{
      \hypo{\ldots}
      \hypo{a:\OpMod{\Opn}{A}}
      \hypo{b:Ba}
      \infer3{
        \brk{\Opn\hookrightarrow a\mid b} : \Compr{\brk{\Opn\hookrightarrow x:A\mid Bx}}{\Opn\hookrightarrow a}
      }
    }
    \and
    \ebrule[elimination]{
      \hypo{\ldots}
      \hypo{u:\brk{\Opn\hookrightarrow x:A\mid Bx}}
      \infer2{
        \underline{u} : B u
      }
    }
    \and
    \ebrule[computation]{
      \hypo{\ldots}
      \hypo{a:\OpMod{\Opn}{A}}
      \hypo{b:Ba}
      \infer3{
        \underline{\brk{\Opn\hookrightarrow a\mid b}} = b : B a
      }
    }
    \and
    \ebrule[uniqueness]{
      \hypo{\ldots}
      \hypo{u:\brk{\Opn\hookrightarrow x:A\mid Bx}}
      \infer2{
        u = \brk{\Opn\hookrightarrow u \mid \underline{u}} : \brk{\Opn\hookrightarrow x:A\mid Bx}
      }
    }
  \end{mathpar}
  \caption{The interface to the \emph{refinement type} connective.}
  \label{fig:refinement-type}
\end{figure}

\begin{nota}[Copattern notation]
  We will often define elements of $\brk{\Opn\hookrightarrow x:A\mid B x}$ using
  Agda-style ``copattern'' notation; for instance, the pair of declarations on
  the left is meant to denote the single declaration on the right:
  \[
    \begin{bmatrix*}[l]
      {\Opn}\hookrightarrow{u \eqdef a}\\
      \underline{u}\eqdef b
    \end{bmatrix*}
    \rightsquigarrow
    \prn{u \eqdef \brk{\Opn\hookrightarrow a \mid b}}
  \]
\end{nota}

\subsection{The state monad in \texorpdfstring{\iGDTTRefLRAT{}}{iGDTT/ref/lrat}}

In \cref{sec:dependent:axiomatization}, we extended \iGDTT{} with general
reference types and a state monad $\TpT*$. In this section, we add to the
axiomatization of \iGDTTRefLRAT{} constructs governing the weak bisimulation of
stateful computations. The following two rules make it possible to exhibit a
correspondence between two computations that take different numbers of steps:
\begin{mathpar}
  \ebrule{
    \infer0{
      \TmStep\SubL  : \Compr{\TpT{\TpUnit}}{\OpnL\hookrightarrow \TmStep, \OpnR\hookrightarrow \TmRet\TmUnit}
    }
  }
  \and
  \ebrule{
    \infer0{
      \TmStep\SubR  : \Compr{\TpT{\TpUnit}}{\OpnL\hookrightarrow \TmRet\TmUnit, \OpnR\hookrightarrow \TmStep}
    }
  }
\end{mathpar}

\begin{exa}
  With the above in hand, it is possible exhibit a correspondence between two
  computations that access the memory different numbers of times. Let
  $r:\TpRef{\TpInt}$ and consider the programs $M\SubL,M\SubR$ defined
  below:

  \iblock{
    \mrow{
      M\SubL \eqdef
      \TmBind{x}{\TmGet{\TpInt}{r}}
      \TmSet{\TpInt}{r}{\prn{x+1}};
      \TmBind{y}{\TmGet{\TpInt}{r}}
      \TmSet{\TpInt}{r}{\prn{y+1}};
      \TmRet\TmUnit
    }
    \mrow{
      M\SubR\eqdef
      \TmBind{x}{\TmGet{\TpInt}{r}}
      \TmSet{\TpInt}{r}{\prn{x+2}};
      \TmRet\TmUnit
    }
  }

  We may define a correspondence $\mathcal{M} :
  \Compr{\TpT{\TpUnit}}{\OpnL\hookrightarrow M\SubL,\OpnR\hookrightarrow M\SubR}$ as
  follows:

  \iblock{
    \mrow{
      \mathcal{M} \eqdef
      \TmBind{x}{\TmGet{\TpInt}{r}}
      \TmStep\SubL;
      \TmSet{\TpInt}{r}{\prn{x+2}};
      \TmRet\TmUnit
    }
  }

  While $\mathcal{M}$ is evidently of type $\TpT{\TpUnit}$, it remains to argue that
  it is equal to $M\SubL$ (resp. $M\SubR$) under the assumption $\OpnL$
  (resp. $\OpnR$). Both claims follow from equational reasoning; assuming
  $\OpnL$ we compute:
  \begin{align*}
    \OpnL\hookrightarrow M\SubL
    &=
    \TmBind{x}{\TmGet{\TpInt}{r}}
    \Alert{
      \TmSet{\TpInt}{r}{\prn{x+1}};
      \TmBind{y}{\TmGet{\TpInt}{r}}
    }\,
    \TmSet{\TpInt}{r}{\prn{\Alert{y}+1}};
    \TmRet\TmUnit
    \\
    &=
    \TmBind{x}{\TmGet{\TpInt}{r}}
    \TmStep;
    \Alert{
      \TmSet{\TpInt}{r}{\prn{x+1}};
      \TmSet{\TpInt}{r}{\prn{x+1+1}}
    };
    \TmRet\TmUnit
    \\
    &=
    \TmBind{x}{\TmGet{\TpInt}{r}}
    \TmStep;
    \TmSet{\TpInt}{r}{\prn{\Alert{x+1+1}}};
    \TmRet\TmUnit
    \\
    &=
    \TmBind{x}{\TmGet{\TpInt}{r}}
    \Alert{\TmStep};
    \TmSet{\TpInt}{r}{\prn{x+2}};
    \TmRet\TmUnit
    \\
    &=
    \TmBind{x}{\TmGet{\TpInt}{r}}
    \TmStep\SubL;
    \TmSet{\TpInt}{r}{\prn{x+2}};
    \TmRet\TmUnit
    \\
    &=
    \mathcal{M}
  \end{align*}

  Likewise, assuming $\OpnR$ we have the following:
  \[
    \OpnR\hookrightarrow M\SubR
    =
    \TmBind{x}{\TmGet{\TpInt}{r}}
    \TmSet{\TpInt}{r}{\prn{x+2}};
    \TmRet\TmUnit
    =
    \TmBind{x}{\TmGet{\TpInt}{r}}
    \TmStep\SubL;
    \TmSet{\TpInt}{r}{\prn{x+2}};
    \TmRet\TmUnit
    =
    \mathcal{M}
  \]
\end{exa}

\subsection{Case study: local references and closures}\label{sec:case-study:local-references}

In this section we consider a synthetic version of an example from Birkedal,
St\o{}vring, and Thamsborg~\cite[\S6.3]{birkedal-stovring-thamsborg:2010}
concerning a correspondence between two imperative counter modules constructed
using a local reference and two closures. Whereas \opcit needed to work inside
the model, the combination of dependent types and synthetic relational
constructs enables us to work directly in the type theory. Consider the
following two implementations of an imperative counter module using local
references, one of which counts up and the other counts down:

\iblock{
  \mrow{
    M\SubL,M\SubR : \TpT\prn{\TpT{\TpUnit} \times \TpT{\TpInt}}
  }

  \mrow{
    M\SubL \eqdef
    \prn{
      \TmBind{r}{\TmNew{\TpInt}{0}}
      \TmRet\prn{
        \Con{incr}\,r,
        \TmGet{\TpInt}{r}
      }
    }
  }

  \mrow{
    M\SubR \eqdef
    \prn{
      \TmBind{r}{\TmNew{\TpInt}{0}}
      \TmRet\prn{
        \Con{decr}\,r,
        \TmBind{x}{\TmGet{\TpInt}{r}}
        \TmRet\prn{-x}
      }
    }
  }

  \row

  \mrow{
    \Con{incr},\Con{decr} : \TpRef{\TpInt}\to \TpT{\TpUnit}
  }
  \mrow{
    \Con{incr}\,r \eqdef
    \TmBind{x}{\TmGet{\TpInt}{r}}
    \TmSet{\TpInt}{r}{\prn{x+1}}
  }
  \mrow{
    \Con{decr}\,r \eqdef
    \TmBind{x}{\TmGet{\TpInt}{r}}
    \TmSet{\TpInt}{r}{\prn{x-1}}
  }
}

We wish to construct a correspondence between $M\SubL$ and $M\SubR$, \ie an
element
$\mathcal{M}:\Compr{\TpT\prn{\TpT{\TpUnit}\times\TpT{\TpInt}}}{\OpnL\hookrightarrow
M\SubL,\OpnR\hookrightarrow M\SubR}$. To do so, we define a correspondence on
$\TpInt$ describing the heap invariant using a refinement type:

\iblock{
  \mrow{
    \mathcal{Z} : \Compr{\TpSet}{\Opn\hookrightarrow \TpInt}
  }
  \mrow{
    \mathcal{Z} \eqdef
    \brk{
      \Opn\hookrightarrow x : \TpInt\mid
      \ClMod{\Opn}{
        \Compr{
          \prn{
            i\SubL : \Compr{\TpInt}{\OpnL\hookrightarrow x},
            i\SubR : \Compr{\TpInt}{\OpnR\hookrightarrow x}
          }
        }{
          i\SubL = -i\SubR
        }
      }
    }
  }
}

\begin{exegesis}
  It is important to understand why the invariant $\mathcal{Z}$ is defined the way it is.
  \begin{enumerate}

    \item First we attach an element $x:\OpMod{\Opn}{\TpInt}$, which is in
      particular a pair of a ``left'' element $x\SubL
      :\OpMod{\OpnL}{\TpInt}$ and a ``right'' element
      $x\SubR:\OpMod{\OpnR}{\TpInt}$ with no conditions whatsoever.

    \item We wish to assert that $x\SubR$ is the negation of $x\SubL$; but this
      is not type correct, because these partial integers lie in disjoint
      worlds.

    \item Instead we will glue onto the partial integer $x =
      \brk{\OpnL\hookrightarrow x\SubL,\OpnR\hookrightarrow x\SubR}$ a pair of
      total integers $i\SubL,i\SubR:\TpInt$ restricting under
      $\OpnL,\OpnR$ to $x\SubL,x\SubR$ such that $i\SubL= -i\SubR$.

    \item It is necessary to hide the attachment of $\prn{i\SubL,i\SubR}$ to
      $x$ underneath the closed modality $\ClMod{\Opn}{-}$, since otherwise
      $\mathcal{Z}$ would restrict to strictly more than just $\TpInt$ under $\Opn$. Without the closed modality, the formation rule for the refinement type employed here would not apply.
  \end{enumerate}
\end{exegesis}

With the invariant $\mathcal{Z}$ in hand, we will construct a \emph{new}
counter implementation $\mathcal{M}$ that allocates an element of
$\mathcal{Z}$; this new counter implementation should restrict ``on the left''
to $M\SubL$ and ``on the right'' to $M\SubR$. We divide the construction into
three main subroutines $\Con{init}\Sub{\mathcal{M}}$,
$\Con{tick}\Sub{\mathcal{M}}$, and $\Con{read}\Sub{\mathcal{M}}$:

\iblock{
  \mrow{
    \mathcal{M} : \Compr{\TpT\prn{\TpT{\TpUnit}\times \TpT{\TpInt}}}{
      \OpnL\hookrightarrow M\SubL,
      \OpnR\hookrightarrow M\SubR,
    }
  }

  \mrow{
    \mathcal{M} \eqdef
    \prn{
      \TmBind{r}{\TmNew{\mathcal{Z}}{\Con{init}\Sub{\mathcal{M}}}}
      \TmRet\prn{
        \Con{tick}\Sub{\mathcal{M}}\,r,
        \Con{read}\Sub{\mathcal{M}}\,r
      }
    }
  }
}

For the initialization, we simply construct an element of $\mathcal{Z}$ witnessing $0=-0$.

\iblock{
  \mrow{
    \Con{init}\Sub{\mathcal{M}} : \Compr{\mathcal{Z}}{\Opn\hookrightarrow 0}
  }
  \mrow{
    {\Opn}\hookrightarrow{\Con{init}\Sub{\mathcal{M}} \eqdef 0}
  }
  \mrow{
    \underline{\Con{init}\Sub{\mathcal{M}}} \eqdef \eta\Sub{\ClMod{\Opn}}\prn{0,0}
  }
}

Next we describe the operation that ticks the counter, factoring through an
auxiliary operation $\Con{upd}\Sub{\mathcal{M}}$ that witnesses the
preservation of the $\mathcal{Z}$ correspondence by the change to the reference
cell.

\iblock{
  \mrow{
    \Con{tick}\Sub{\mathcal{M}} : \Compr{\TpRef{\mathcal{Z}}\to \TpT{\TpUnit}}{\OpnL\hookrightarrow\Con{incr},\OpnR\hookrightarrow\Con{decr}}
  }
  \mrow{
    \Con{tick}\Sub{\mathcal{M}} r \eqdef
    \TmBind{x}{\TmGet{\mathcal{Z}}{r}}
    \TmSet{\mathcal{Z}}{r}{
      \prn{\Con{upd}\Sub{\mathcal{M}} x}
    }
  }

  \row

  \mrow{
    \Con{upd}\Sub{\mathcal{M}} : \Compr{\mathcal{Z}\to\mathcal{Z}}{\OpnL\hookrightarrow \lambda x.x+1,\OpnR\hookrightarrow\lambda x.x-1}
  }
  \mrow{
    {\OpnL}\hookrightarrow{\Con{upd}\Sub{\mathcal{M}} x \eqdef x+1}
  }
  \mrow{
    {\OpnR}\hookrightarrow{\Con{upd}\Sub{\mathcal{M}} x \eqdef x-1}
  }
  \mrow{
    \underline{\Con{upd}\Sub{\mathcal{M}} x} \eqdef
    \TmBind{\prn{i\SubL,i\SubR}}{\underline{x}}
    \eta\Sub{\ClMod{\Opn}}\prn{
      i\SubL +1,
      i\SubR -1
    }
  }
}

In the definition of $\underline{\Con{upd}\Sub{\mathcal{M}}x}$, we have used
the Kleisli extension for the closed modality $\ClMod{\Opn}{-}$. The
well-typedness of our definition follows from the fact that $i\SubL=-i\SubR$
implies $i\SubL+1=-\prn{i\SubR-1}$.
Finally we
construct a function to read the contents of the local state, factoring through
a function $\Con{tidy}\Sub{\mathcal{M}}$ that verifies the correspondence
between $x$ and $-x$ in the output:

\iblock{
  \mrow{
    \Con{read}\Sub{\mathcal{M}} : \Compr{\TpRef{\mathcal{Z}}\to \TpT{\TpInt}}{\OpnL\hookrightarrow \lambda r.\TmGet{\TpInt}{r}, \OpnR\hookrightarrow \lambda r.\TmBind{x}{\TmGet{\TpInt}{r}} \TmRet\prn{-x}}
  }
  \mrow{
    \Con{read}\Sub{\mathcal{M}} r \eqdef
    \TmBind{x}{\TmGet{\mathcal{Z}}{r}}
    \TmRet\prn{
      \Con{tidy}\Sub{\mathcal{M}}x
    }
  }

  \row

  \mrow{
    \Con{tidy}\Sub{\mathcal{M}} : \Compr{\mathcal{Z}\to\TpInt}{\OpnL\hookrightarrow \lambda x. x, \OpnR\hookrightarrow \lambda x. {-x}}
  }

  \mhang{
    \Con{tidy}\Sub{\mathcal{M}} x \eqdef
  }{
    \mhang{
      \Kwd{case}\ \underline{x}\ \Kwd{of}
    }{
      \mrow{
        \eta\Sub{\ClMod{\Opn}}\prn{i\SubL,i\SubR}\hookrightarrow i\SubL
      }
      \mrow{
        \star \hookrightarrow \brk{\OpnL\hookrightarrow x,\OpnR\hookrightarrow -x}
      }
    }
  }
}

\paragraph{Bisimulation of denotations}
By interpreting the correspondence $\mathcal{M}$ in our semantic model (see
\cref{sec:semantics}), it will follow that the results of initializing
$M\SubL,M\SubR$ with any heap are \emph{weakly bisimilar}, \ie equal up
to computation steps. In fact, this particular example exhibits a \emph{strong}
bisimulation.

\subsection{Case study: correspondences in abstract data types}\label{sec:case-study:adt}

We may adapt our previous example to exhibit a simulation between two
implementations of an abstract data type for imperative counters as in
Birkedal, St\o{}vring, and Thamsborg~\cite[\S6.2]{birkedal-stovring-thamsborg:2010}. In particular, consider the
following abstract data type:

\iblock{
  \mrow{
    \Con{COUNTER} \eqdef \Exists{\alpha:\TpSet} \TpT\alpha \times \prn{\alpha\to \TpT{\TpUnit}} \times \prn{\alpha\to \TpT{\TpInt}}
  }
}

We consider the following two implementations of $\Con{COUNTER}$:

\iblock{
  \mrow{
    M\SubL \eqdef
    \Con{pack}\,\prn{
      \TpRef{\TpInt}, \prn{\TmNew{\TpInt}{0}, \Con{incr}, \lambda r. \TmGet{\TpInt}{r}}
    }
  }
  \mrow{
    M\SubR \eqdef
    \Con{pack}\,\prn{
      \TpRef{\TpInt}, \prn{\TmNew{\TpInt}{0}, \Con{decr}, \lambda r. \TmBind{x}{\TmGet{\TpInt}{r}} \TmRet\prn{-x}}
    }
  }
}

We may construct a correspondence between these two imperative counter
structures as follows, re-using the constructions of
\cref{sec:case-study:local-references}.

\iblock{
  \mrow{
    \mathcal{M} : \Compr{\Con{COUNTER}}{\OpnL\hookrightarrow M\SubL,\OpnR\hookrightarrow M\SubR}
  }
  \mrow{
    \mathcal{M} = \Con{pack}\,\prn{
      \TpRef{\mathcal{Z}}, \prn{
        \TmNew{\mathcal{Z}}{\Con{init}\Sub{\mathcal{M}}},
        \Con{tick}\Sub{\mathcal{M}},
        \Con{read}\Sub{\mathcal{M}}
      }
    }
  }
}

\paragraph{Bisimulation of denotations}
As in \cref{sec:case-study:local-references}, by interpreting the
correspondence $\mathcal{M}$ into the semantic model of \cref{sec:semantics} we
may exhibit a weak (in fact, strong) bisimulation between the denotations of
$M\SubL$ and $M\SubR$ when initialized with any heap.

\section{Semantic models of \texorpdfstring{\iGDTT}{iGDTT} and \texorpdfstring{\iGDTTRefLRAT}{iGDTT/ref/lrat}}\label{sec:semantics}

Thus far we have introduced a series of type theories---\iGDTT{}, \iGDTTRef{},
and \iGDTTRefLRAT{}---but we have not yet proven that these type theories are
consistent. In this section we justify all three theories by constructing
semantic models.

In \cref{sec:semantics:base-model} we show to build a model of \iGDTT{} based
combining realizability and the topos of trees. In
\cref{sec:semantics:hofmann-streicher}, we show that presheaves internal to a
model of \iGDTT{} also assemble into a model of \iGDTT{}. Using this, we then
lift a base model to a model of \iGDTTRefLRAT{} in \cref{sec:semantics:kripke}
by taking presheaves on a carefully chosen category and constructing a more
sophisticated version of the possible-worlds model introduced in
\cref{sec:easy}.
It proves most convenient to pass directly from a model of \iGDTT{} to a model
of \iGDTTRefLRAT{} rather than constructing an intermediate model of
\iGDTTRef{}. Of course, any model of \iGDTTRefLRAT{} is (by restriction) a
model of \iGDTTRef{}.

The construction of this sequence of models immediately yields the consistency
of all three type theories. The particular models, however, show more than
this. We use our model constructions to argue that a relation constructed in
\iGDTTRefLRAT{} induces a weak bisimulation between the related programs in the
induced model of \iGDTTRef{}.

\begin{rem}
  Unlike the rest of the paper, throughout this section we assume knowledge of
  category theory. In particular, in addition to the basic concepts of category
  theory we assume some familiarity with categorical realizability (assemblies,
  modest sets, \etc). We refer the reader to Van Oosten~\cite{van-oosten:2008}
  or Streicher~\cite{streicher:2017-2018} for a thorough introduction.
\end{rem}

\subsection{Constructing base models of \texorpdfstring{\iGDTT}{iGDTT}}\label{sec:semantics:base-model}

A concrete model of \iGDTT{} can be constructed using a combination of
realizability and presheaves. Rather than belabor the already well-studied
interpretation of dependent type theory into well-adapted categories (\eg
locally cartesian closed categories, topoi, \etc{})~\cite{hofmann:1997}, we
describe and exhibit the relevant categorical structure needed to interpret the
nonstandard aspects of \iGDTT{}: impredicative universes and guarded recursion.

Let $\SCat$ be a realizability topos and let $\prn{\mathbb{O},\leq,\prec}$ be a
\DefEmph{separated intuitionistic well-founded poset} in $\SCat$, \ie a poset
$\prn{\mathbb{O},\leq}$ equipped with a transitive subrelation
${\prec}\subseteq{\leq}\subseteq\mathbb{O}\times\mathbb{O}$ satisfying the
following additional conditions:
\begin{enumerate}
  \item \emph{separation}: the object $\mathbb{O}$ is an assembly and both ${\leq},{\prec}$ are regular subobjects,
  \item \emph{left compatibility}: if $u\leq v$ and $v\prec w$ then $u \prec w$,
  \item \emph{right compatibility}: if $u \prec v$ and $v\leq w$ then $u\prec w$,
  \item \emph{well-foundedness}: every element of $\mathbb{O}$ is
    $\prec$-accessible, where the $\prec$-accessible elements are defined to be
    the smallest subset $I\subseteq \mathbb{O}$ such that if $v\in I$ for all
    $v\prec u$, then $u\in I$.
\end{enumerate}

\begin{rem}
  In classical mathematics, the appropriate notion of well-founded poset
  is considerably simpler as we tend to define $x \prec y \Leftrightarrow x\leq y
  \land x \not= y$. This simpler style of well-founded order is particularly
  inappropriate for intuitionistic mathematics (\eg the mathematics of a
  realizability topos), in which $\mathbb{O}$ need not have
  decidable equality. Our definition of intuitionistic well-founded posets is
  inspired by Taylor's analysis of intuitionistic ordinals~\cite{taylor:1996}.
\end{rem}

Given a hierarchy of Grothendieck universes $\mathscr{U}_i$ in $\SET$ such that
$\Gamma\mathbb{O}\in\mathscr{U}_0$, we conclude:

\begin{thm}\label{thm:igdtt-model}
  The category of internal diagrams\footnote{We refer to Borceux~\cite[\S8.1]{borceux:1994:vol1} for discussion of internal diagrams.}
  $\ECat = \brk{\OpCat{\mathbb{O}},\SCat}$ is a
  model of \iGDTT{} in which:
  \begin{enumerate}

    \item the predicative universes $\TpType_i$ are modelled by the
      Hofmann--Streicher liftings~\cite{hofmann-streicher:1997} of the universes of $\mathscr{U}_i$-small assemblies in
      $\SCat$,

    \item the impredicative universes $\TpProp,\TpSet\in\TpType_i$ is modelled
      by the Hofmann--Streicher liftings of the universes of $\lnot\lnot$-closed
      propositions and of modest sets in $\SCat$ respectively,

    \item the later modality $\Ltr$ is computed explicitly by the equation $\prn{\Ltr{A}}u = \Lim{\mathrlap{v\prec u}} Av$.

  \end{enumerate}
\end{thm}

\begin{proof}
  The guarded recursive aspects follow from the general results of
  Palombi and Sterling~\cite{palombi-sterling:2023}; for the rest, it can be seen that the
  Hofmann--Streicher lifting of a pair of universes preserves impredicativity
  of the lower universe when $\mathbb{O}$ is internal to the upper universe.
  Our interpretation of $\TpProp$ is automatically univalent and satisfies
  proof irrelevance, and the only subtle point is that it is closed under
  equality. This follows because every type classified by
  either $\TpSet$ or $\TpType$ is $\lnot\lnot$-separated (since they are
  assemblies), which means precisely that their equality predicates are
  $\lnot\lnot$-closed.
\end{proof}

We describe the ``standard'' instantiation of \cref{thm:igdtt-model} in \cref{ex:eff-psh} below.

\begin{exa}[Standard model of \iGDTT{}]\label{ex:eff-psh}
  Let $\Kwd{Eff}$ be the \emph{effective topos}~\cite{hyland:1982}, namely the
  realizability topos constructed from Kleene's first algebra; there is a
  modest intuitionistic well-founded poset $\omega$ in $\Kwd{Eff}$ given by the
  natural numbers object of the topos under inequality and strict inequality.
  Then $\brk{\OpCat{\omega},\Kwd{Eff}}$ is
  model of \iGDTT{} according to \cref{thm:igdtt-model}.
\end{exa}

From the non-emptiness of $\omega$ it follows immediately that
$\brk{\OpCat{\omega},\Kwd{Eff}}$ is a non-trivial topos, \ie one in which the
initial object is not inhabited. Thus we obtain the following corollary:

\begin{cor}\label{cor:igdtt-consistent}
  \iGDTT{} is consistent.
\end{cor}

\begin{rem}[Strong and weak completeness]
  It is important that each $\TpType_i$ be built from a subuniverse of small \emph{assemblies}
  rather than of more general types in $\SCat$: the full internal subcategory spanned by modest sets
  is strongly complete over assemblies, but only \emph{weakly} complete over the rest of
  $\SCat$~\cite{hyland-robinson-rosolini:1990}. Weak completeness is insufficient to interpret
  $\ForallSymbol$.
\end{rem}

\subsection{The Hofmann--Streicher lifting of a model of \texorpdfstring{\iGDTT}{iGDTT}}\label{sec:semantics:hofmann-streicher}

Let $\ECat$ be an elementary topos model of \iGDTT{} \eg{}, \cref{ex:eff-psh}; we
will show that for certain internal categories $\CICat$ in $\ECat$, the topos of
internal diagrams $\brk{\CICat,\ECat}$ is a model of \iGDTT{}. To avoid
confusion, we will write things like $\TpSet,\TpType_i$ for constructs of
$\ECat$ and $\Ol{\TpSet},\Ol{\TpType_i}$ for the corresponding constructs that
we will define in the category of internal diagrams $\brk{\CICat,\ECat}$.

Let $\CICat$ be a category internal to $\TpType_0$, \ie a category object in
$\ECat$ whose object of objects and hom objects are classified by the universe
$\TpType_0$. Then following
Hofmann and Streicher~\cite{hofmann-streicher:1997} we may lift each universe
$\mathbb{X}\in\brc{\TpProp,\TpSet,\TpType_i}$ of $\ECat$ to a universe
$\Ol{\mathbb{X}}$ in $\brk{\CICat,\ECat}$, setting $\Ol{\mathbb{X}}c =
\brk{c\downarrow\CICat,\mathbb{X}}$.
The only surprise is that $\Ol{\TpSet}$ remains classified by $\Ol{\TpType_i}$,
considering the fact that the coslice $c\downarrow\CICat$ is large relative to
$\TpSet$. Nevertheless we have $\Ol{\TpSet}\in\Ol{\TpType_i}$ because $\TpSet$
is impredicative in $\TpType_i$; by an explicit computation, it follows that
$\Ol{\TpSet}$ can be encoded by an element of $\Ol{\TpType_i}$ using only
$\Forall{}$, $\Sum{}$, and $\prn{=}$ in each fiber. The impredicativity of
$\TpSet$ in $\TpType_i$ likewise ensures that $\Ol{\TpSet}$ is impredicative
over $\Ol{\TpType_i}$.

Closure under the constructs of guarded dependent type theory then follows once
again from the general results of Palombi and
Sterling~\cite{palombi-sterling:2023} concerning the pointwise lifting of
guarded recursion from the base into presheaves. It can also be seen that good
properties of the later modality are preserved by this lifting; if $\ECat$ is
\DefEmph{globally adequate} in the sense of \opcit and $\CICat$ has an initial
object, then the $\brk{\CICat,\ECat}$ is also globally adequate. Global
adequacy means that the global points of $\Ltr{\mathbb{N}}$ are the actual
natural numbers. Summarizing:

\begin{lem}\label{lem:igdtt-model-hs-lifting}
  Let $\ECat$ be an elementary topos model of \iGDTT{} and let $\CICat$ be a category internal to $\TpType_0\in\ECat$.
  The category of internal diagrams $\brk{\CICat,\ECat}$ is a model of \iGDTT{} in which:
  \begin{enumerate}

    \item all universes are modelled by the
      Hofmann--Streicher liftings~\cite{hofmann-streicher:1997} of the corresponding universes in $\ECat$;

    \item the later modality $\Ltr$ is computed pointwise, \ie we have
      $\prn{\Ltr{A}}c=\Ltr\prn{Ac}$.

  \end{enumerate}
\end{lem}

\paragraph{Strict extension structures}
In order to construct the store model in \cref{sec:semantics:kripke}, we will
require that the impredicative universe $\TpSet$ also supports a \emph{strict
  extension structure}. While we require this property in only one concrete
setting, it is most natural to define with respect to an arbitrary
\emph{dominance} $\Mor|>->|{\ObjTerm{\ECat}}{\Sigma}\in\ECat$, \ie a universe of
propositions closed under truth and dependent conjunction.

Recall that for any $\phi : \Sigma$, a partial element $A :
\OpMod{\phi}{\TpSet}$ can be extended to a total element $B : \TpSet$ that
restricts to $A$ up to isomorphism \eg{}, by setting $B$ to $\OpMod{\phi}{A}$.
In \cref{sec:semantics:kripke}, however, we require the ability to choose an
extension which agrees exactly with $A$ under $\phi$. We say a universe
$\mathbb{X}$ has a \DefEmph{strict extension structure} when this is possible \ie{},
when it validates following rules:

\begin{mathpar}
  \ebrule[strict extension]{
    \hypo{\phi:\Sigma}
    \hypo{A : \OpMod{\phi}{\mathbb{X}}}
    \infer2{
      \phi_*A : \mathbb{X}
    }
  }
  \and
  \ebrule[strict extension intro]{
    \hypo{\phi:\Sigma}
    \hypo{A:\OpMod{\phi}{\mathbb{X}}}
    \hypo{a:\OpMod{\phi}{A}}
    \infer3{
      \gl{\phi}a : \phi_*A
    }
  }
  \\
  \ebrule{
    \hypo{\phi:\Sigma}
    \hypo{A:\OpMod{\phi}{\mathbb{X}}}
    \hypo{\_ : \phi}
    \infer3{
      \phi_*A \equiv A : \mathbb{X}
    }
  }
  \and
  \ebrule{
    \hypo{\phi:\Sigma}
    \hypo{A:\OpMod{\phi}{\mathbb{X}}}
    \hypo{a:\OpMod{\phi}{A}}
    \hypo{\_:\phi}
    \infer4{
      \gl{\phi}a \equiv a : A
    }
  }
\end{mathpar}

Fortunately, given a dominance $\Sigma$ we are always able to replace a
universe $\mathbb{X}$ with a new universe equipped with an extension structure
without perturbing the collection of types it classifies
(\cref{lem:strictification}). The idea is to replace $\mathbb{X}$ with its
$\Sigma$-partial map classifier.

\begin{defi}[Partial element classifier]
  Let $A$ be a type and let $\Sigma$ be a dominance; the
  \DefEmph{$\Sigma$-partial map classifier} $A^+$ is defined to be the
  dependent sum $\Sum{\phi:\Sigma}{\OpMod{\phi}{A}}$. We have a natural transformation
  $\Mor[\eta]{\ArrId}{\prn{-}^+}$ sending each $a:A$ to the pair
  $\prn{\top,a}:A^+$.
\end{defi}

For any universe $\mathbb{X}$ that classifies each $\phi:\Sigma$, we have an
algebra structure $\Mor[\Pi]{\mathbb{X}^+}{\mathbb{X}}$ for the raw endofunctor
$-^+$, sending each $A:\mathbb{X}^+$ to its dependent product
$\Pi{Q}=\OpMod{\pi_1Q}{\pi_2Q}$. Viewing $\mathbb{X}$ as an internal groupoid,
we see that $\Pi$ is a ``pseudo-algebra'' structure for the \textbf{pointed}
endofunctor $-^+$, as each $\Pi\prn{\top,A}$ is isomorphic to $A$.\footnote{A
similar observation is made by Escard\'o~\cite{escardo:2021} in his
investigation of injective types in univalent mathematics.}  It can be seen
that a strict extension structure for $\mathbb{X}$ is exactly the same thing as
a strictification of $\Pi$, \ie a strict algebra for the pointed endofunctor
$-^+$ that is additionally isomorphic to $\Pi$.\footnote{Algebras for the
pointed endofunctor $-^+$ play an important role in the recent work of Awodey
on Quillen model structures in cubical sets~\cite{awodey:2021:qms}, from whom
we have borrowed some of our notation.}

\begin{lem}[Strictification]\label{lem:strictification}
  Let $\Sigma$ be a dominance such that each $\phi:\Sigma$ is classified by
  $\mathbb{X}$; then there exists a universe $\mathbb{X}^s$ equipped with a
  strict extension structure for $\Sigma$ such that $\mathbb{X}$ and $\mathbb{X}^s$
  are \DefEmph{strongly equivalent} internal categories.
\end{lem}

\begin{proof}
  We set $\mathbb{X}^s$ to be the $\Sigma$-partial element classifier
  $\mathbb{X}^+$ itself. The strict algebra structure
  $\Mor[\alpha]{\mathbb{X}^{++}}{\mathbb{X}^+}$ is given by the multiplication
  operation of the partial element classifier monad, \ie the dependent conjunction of the dominance $\Sigma$.
  The strictness of the algebra structure follows from the fact that $\Sigma$
  is univalent, like any dominance. The equivalence
  $\Mor{\mathbb{X}^+}{\mathbb{X}}$ is given by the pseudo-algebra $\Pi$.
\end{proof}

We therefore combine \cref{lem:igdtt-model-hs-lifting} with
\cref{lem:strictification} to obtain the following theorem.

\begin{thm}\label{thm:igdtt-model-lifting}
  Let $\ECat$ be an elementary topos model of \iGDTT{} and let $\CICat$ be a category internal to $\TpType_0\in\ECat$.
  The category of internal diagrams $\brk{\CICat,\ECat}$ is a model of \iGDTT{} where:
  \begin{enumerate}

  \item the universes $\TpProp,\TpType_i$ are modelled by Hofmann--Streicher
    lifting;

  \item the impredicative universe $\TpSet\in\TpType_i$ is modelled by a
    universe strongly equivalent to the Hofmann--Streicher lifting of
    $\TpSet : \ECat$, but equipped with a strict extension structure for
    $\TpProp$;

  \item the later modality $\Ltr$ is computed pointwise, \ie we have
    $\prn{\Ltr{A}}c=\Ltr\prn{Ac}$.

  \end{enumerate}
\end{thm}

We will use this result twice: first to extend models of \iGDTT{} with
constructs for relational reasoning and parametricity
(\cref{sec:semantics:lrat}), and second to extend models of \iGDTT{} with
higher-order store (\cref{sec:semantics:kripke}). The culmination of this
process is a standard model for \iGDTTRefLRAT{} that modularly combines
realizability and presheaves.

\NewDocumentCommand\SPAN{m}{\mathbf{Span}\,#1}

\subsection{Logical relations as types in the span model of \texorpdfstring{\iGDTT}{iGDTT}}\label{sec:semantics:lrat}

Prior to the modeling \iGDTTRefLRAT, we expose the construction of the
\Alert{span model} of \iGDTT{} on top of a base model. This
intermediate model does not possess the necessary structure to interpret
references, but it does support the relational primitives of \iGDTTRefLRAT{}. We
will eventually use \iGDTT{} itself as an internal language within this model in
\cref{sec:semantics:kripke} to build the store model.

\subsubsection{The span model} Given a topos model $\ECat$ of \iGDTT{} we may
take the topos $\SPAN{\ECat}$ of spans in $\ECat$, which can be constructed as a
category of diagrams $\brk{\mathbb{C},\ECat}$ where $\mathbb{C}$ is the generic
span $\brc{\mathsf{L} \leftarrow \mathsf{C} \rightarrow \mathsf{R}}$ viewed as
an internal poset in $\TpSet\in\ECat$. When $\ECat$ is the category of
presheaves on a well-founded poset $\mathbb{O}$ computed in a realizability
topos such as $\Kwd{Eff}$ as in \cref{sec:semantics:base-model}, we obtain from
\cref{thm:igdtt-model-lifting} a new model of \iGDTT{} in $\SPAN{\ECat}$ in
which all guarded constructs are computed pointwise.
The span model $\SPAN{\ECat}$ contains propositions $\OpnL,\OpnR$ given by the
representables
$\Yo{\OpCat{\mathbb{C}}}{\mathsf{L}},\Yo{\OpCat{\mathbb{C}}}{\mathsf{R}}$
respectively. The span model is moreover closed under the \emph{refinement}
type connective $\brk{\Opn\hookrightarrow x:A\mid Bx}$ where
$\Opn=\OpnL+\OpnR$ is the (disjoint) union as usual; this follows as
the subterminals $\OpnL,\OpnR$ have pointwise decidable image in
$\ECat$~\cite{orton-pitts:2016}.%

\subsubsection{Modal geometry of the span model}

We may further unravel the open and closed modalities of $\SPAN{\ECat}$ to
provide additional computational and geometric intuition. First we comment that
the propositions $\OpnL,\OpnR,\Opn$ are concretely realized by the spans
$\brc{\ObjTerm{\ECat}\leftarrow \ObjInit{\ECat} \rightarrow \ObjInit{\ECat}}$,
$\brc{\ObjInit{\ECat}\leftarrow \ObjInit{\ECat} \rightarrow \ObjTerm{\ECat}}$,
and $\brc{\ObjTerm{\ECat}\leftarrow \ObjInit{\ECat} \rightarrow
\ObjTerm{\ECat}}$ respectively.

\paragraph{Characterization of open modalities}

For any subterminal object $\Mor|>->|{\phi}{\ObjTerm{\SPAN{\ECat}}}$, the subcategory of $\SPAN{\ECat}$ spanned by $\phi$-transparent
objects can be expressed as the slice $\SPAN{\ECat}\downarrow \phi$. But in the
case of the specific propositions that we have distinguished, more explicit
computations are available:

\begin{enumerate}

  \item A $\Opn$-transparent object is one that is isomorphic to a product span
    $\brc{X\SubL\leftarrow X\SubL\times X\SubR \rightarrow X\SubR}$.

  \item A $\OpnL$-transparent object is one that is isomorphic to a digram of
    the form $\brc{X\SubL\leftarrow X\SubL\rightarrow \ObjTerm{\ECat}}$.

  \item A $\OpnR$-transparent object is one that is isomorphic to a digram of
    the form $\brc{\ObjTerm{\ECat}\leftarrow X\SubR\rightarrow X\SubR}$.
\end{enumerate}

Thus we compute the open modalities explicitly:
\begin{enumerate}

  \item The open modality $\prn{\OpMod{\Opn}{-}}$ takes a span
    $\brc{X\SubL\leftarrow \tilde{X}\rightarrow X\SubR}$
    to the span $\brc{X\SubL\leftarrow X\SubL\times X\SubR\rightarrow X\SubR}$.

  \item The open modality $\prn{\OpMod{\OpnL}{-}}$ takes a span
    $\brc{X\SubL\leftarrow \tilde{X}\rightarrow X\SubR}$
    to the span $\brc{X\SubL\leftarrow X\SubL\rightarrow \ObjTerm{\ECat}}$.

  \item The open modality $\prn{\OpMod{\OpnR}{-}}$ takes a span
    $\brc{X\SubL\leftarrow \tilde{X}\rightarrow X\SubR}$
    to the span $\brc{\ObjTerm{\ECat}\leftarrow X\SubR\rightarrow X\SubR}$.
\end{enumerate}

Based on the above, it is easy to see (\eg) that $\OpMod{\Opn}{A}$ is
isomorphic to $\prn{\OpMod{\OpnL}{A}}\times \prn{\OpMod{\OpnR}{A}}$.

\paragraph{Characterization of closed modality}

Here we give an analogous discussion of the closed modality.  A $\Opn$-sealed
object is one that is isomorphic to a diagram of the form
$\brc{\ObjTerm{\ECat}\leftarrow \tilde{X}\rightarrow\ObjTerm{\ECat}}$.  Then
the closed modality $\prn{\ClMod{\Opn}{-}}$ takes a span $\brc{X\SubL\leftarrow
\tilde{X}\rightarrow X\SubR}$ to the span
$\brc{\ObjTerm{\ECat}\leftarrow\tilde{X}\rightarrow \ObjTerm{\ECat}}$.

\subsubsection{Synthetic weak bisimulation}\label{sec:synthetic-weak-bisimulation}

In the span model of \iGDTT{}, we may define a new type connective $\TpBisim* :
\Compr{\TpSet \to \TpSet}{\Opn\hookrightarrow \TpL*}$ that expresses a weak bisimulation between two computations;
in other words, we will have $\OpMod{\Opn}{\TpBisim{A} =
\TpL{A}}$ and hence both $\OpMod{\OpnL}{\TpBisim{A} = \TpL{A}}$
and $\OpMod{\OpnR}{\TpBisim{A} = \TpL{A}}$. Given
$u\SubL:\OpMod{\OpnL}{\TpL{A}}$ and $u\SubR:\OpMod{\OpnR}{\TpL{A}}$, an
element of $\Compr{\TpBisim{A}}{\OpnL\hookrightarrow
u\SubL,\OpnR\hookrightarrow u\SubR}$ will be a proof that $u\SubL$ and $u\SubR$
are weakly bisimilar. Because $A$ itself can be any type (\eg any synthetic
correspondence), our definition of $\TpBisim{A}$ can be seen to be an
adaptation of the lifting of a correspondence from
M\o{}gelberg and Paviotti~\cite{mogelberg-paviotti:2016}.
We define $\TpBisim{A}$ by solving a guarded recursive domain equation
in $\TpSet$:

\iblock{
  \mrow{
    \TpBisim{A} \eqdef \brk{\Opn\hookrightarrow u:\TpL{A} \mid \Con{WeaklyBisimilar}_A u}
  }

  \row

  \mhang{
    \Kwd{data}\ \Con{Done}_A : \prn{\OpMod{\Opn}{\TpL{A}}}\to \TpSet\ \Kwd{where}
  }{
    \mrow{
      \Con{stop} : \prn{a:\OpMod{\Opn}{A}}\to \Con{Done}_A\prn{\eta a}
    }
    \mrow{
      \Con{waitR} : \prn{a:\OpMod{\Opn}{A}, n : \mathbb{N}\Sub{\geq1}} \to \Con{Done}_A\,\brk{\OpnL\hookrightarrow \eta a, \OpnR\hookrightarrow\delta^n\eta a}
    }
    \mrow{
      \Con{waitL} : \prn{a:\OpMod{\Opn}{A}, n : \mathbb{N}\Sub{\geq1}} \to \Con{Done}_A\,\brk{\OpnL\hookrightarrow \delta^n \eta a, \OpnR\hookrightarrow \eta a}
    }
  }

  \row

  \mhang{
    \Kwd{quotient}\ \Kwd{data}\ \Con{WeaklyBisimilar}_A : \prn{\OpMod{\Opn}{\TpL{A}}} \to \TpSet\ \Kwd{where}
  }{
    \mrow{
      \Con{done} : \brc{u:\OpMod{\Opn}{\TpL{A}}} \to \Con{Done}_Au \to \Con{WeaklyBisimilar}_Au
    }
    \mrow{
      \Con{step} : \prn{u:\Ltr\TpBisim{A}} \to \Con{WeaklyBisimilar}_A\,\prn{\vartheta\Sub{\TpL{A}}\,u}
    }
    \mrow{
      \star : \brc{\_:\Opn, u:\TpL{A}}\to \Con{WeaklyBisimilar}_A u
    }
    \mrow{
      \_ : \brc{\_:\Opn,u:\TpL{A},p:\Con{WeaklyBisimilar}_A u}\to p = \star
    }
  }
}

We have ensured that $\Con{WeaklyBisimilar}_A$ is valued in $\Opn$-sealed types
by defining it as a \emph{quotient inductive} definition; here again we recall \cref{rem:effectivity}. Thus the use of the
refinement type in $\TpBisim{A}$ is well-defined.
We note that $\TpBisim{A}$ enjoys a similar interface to that of
$\TpL{A}$; in particular, it inherits the structure of a guarded domain from
its components, meaning that it supports recursion:

\iblock{
  \mrow{
    {\Opn}\hookrightarrow{\vartheta\Sub{\TpBisim{A}} u \eqdef \vartheta\Sub{\TpL{A}}\,u}
  }
  \mrow{
    \underline{\vartheta\Sub{\TpBisim{A}} u} \eqdef
    \Con{step}\,\prn{
      \Next*\brk{x\leftarrow u}.\, x
    }
  }
}

Likewise, we may define a unit map $\widetilde{\eta} : A \to \TpBisim{A}$ as follows:

\iblock{
  \mrow{
    {\Opn}\hookrightarrow{\widetilde{\eta}\,a \eqdef \eta\,a}
  }
  \mrow{
    \underline{\widetilde{\eta}\,a} \eqdef \Con{done}\,\prn{\Con{stop}\,a}
  }
}

In order to define the Kleisli extension for $\TpBisim{A}$ we must
first define an operation that adds a step on the left or on the right; we
define the left-handed stepper $\delta\SubL$ below, and treat $\delta\SubR$ symmetrically.

\iblock{
  \mrow{
    \delta\SubL : \TpBisim{A} \to \TpBisim{A}
  }
  \mrow{
    {\Opn}\hookrightarrow{\delta\SubL u \eqdef \brk{\OpnL\hookrightarrow \delta u, \OpnR\hookrightarrow u}}
  }
  \mhang{
    \underline{\delta\SubL u} \eqdef
  }{
    \mhang{
      \Kwd{case}\,\underline{u}\,\Kwd{of}
    }{
      \mrow{
        \Con{done}\,\prn{\Con{stop}\,a} \hookrightarrow \Con{done}\,\prn{\Con{waitL}\,1\,a}
      }
      \mrow{
        \Con{done}\,\prn{\Con{waitL}\,a\,n} \hookrightarrow \Con{done}\,\prn{\Con{waitL}\,\prn{n+1}\,a}
      }
      \mrow{
        \Con{done}\,\prn{\Con{waitR}\,a\,n} \hookrightarrow \Con{step}\,\prn{\Next\prn{\Con{waitR}_?\,\prn{n-1}\,a}}
      }
      \mrow{
        \Con{step}\,u' \hookrightarrow \Con{step}\,\prn{\Next*\brk{x\leftarrow u'}.\, \delta\SubL x}
      }
      \mrow{
        \star \hookrightarrow \brk{\OpnL\hookrightarrow \delta u, \OpnR\hookrightarrow u}
      }
    }
  }

  \row

  \mrow{
    \Con{waitR}_? : \prn{a:\OpMod{\Opn}{A}}\,\prn{n : \mathbb{N}} \to \Con{Done}_A\,\brk{\OpnL\hookrightarrow \eta a, \OpnR\hookrightarrow\delta^n\eta a}
  }
  \mrow{
    \Con{waitR}_?\,a\,0 \eqdef \Con{stop}\,a
  }
  \mrow{
    \Con{waitR}_?\,a\,\prn{n\geq 1} \eqdef \Con{waitR}\,a\,n
  }
}

To see that the case analysis in $\underline{\delta\SubL u}$ is well-defined as
a map out of the quotient, it suffices to observe that all clauses return
$\delta u$ under $\OpnL$ and $u$ under $\OpnR$.
Next we define $\TmStep\SubL,\TmStep\SubR$ to be the generic effects
$\delta\SubL\prn{\TmRet\TmUnit}$ and $\delta\SubR\prn{\TmRet\TmUnit}$
respectively.
We finally define the Kleisli extension for $\TpBisim*$ as well, extending the
Kleisli extension for $\TpL*$:

\iblock{
  \mrow{
    {\Opn}\hookrightarrow{\widetilde{\Con{bind}}\,u\,f \eqdef \Con{bind}\,u\,f}
  }
  \mhang{
    \underline{\widetilde{\Con{bind}}\,u\,f} \eqdef
  }{
    \mhang{
      \Kwd{case}\,\underline{u}\,\Kwd{of}
    }{
      \mrow{
        \Con{done}\,d \hookrightarrow \Con{finalize}\, d\, f
      }
      \mrow{
        \Con{step}\,u' \hookrightarrow
        \Con{step}\,\prn{
          \Next*\brk{x\leftarrow u'}.\,
          \widetilde{\Con{bind}}\, x\, f
        }
      }
      \mrow{
        \star \hookrightarrow \Con{bind}\,u\,f
      }
    }
  }

  \row

  \mrow{
    \Con{finalize}\,\prn{\Con{stop}\,a}\, f \eqdef f a
  }
  \mrow{
    \Con{finalize}\,\prn{\Con{waitL}\,n\,a}\, f \eqdef \delta\SubL^n\, \prn{f a}
  }
  \mrow{
    \Con{finalize}\,\prn{\Con{waitR}\,n\,a}\, f \eqdef \delta\SubR^n\, \prn{f a}
  }
}

\subsection{Higher-order store in the possible worlds model of \texorpdfstring{\iGDTTRefLRAT}{iGDTT/ref/lrat}}\label{sec:semantics:kripke}

We recall the construction of the category $\WICat$ internal to $\TpType_0$ as
well as the objects of heaps $\HICat_w$ from \cref{sec:worlds-and-heaps}. In
this section, we apply \cref{thm:igdtt-model-lifting} to
$\brk{\WICat,\SPAN{\ECat}}$ to construct a dependently typed version of the
store model. In our previous model, the collection of types was global and did
not depend on Kripke worlds; in the new version, the types themselves depend on
heap shapes.  We define the reference type as follows, writing $B\Sub{\vert w'}
: \brk{w'\downarrow\WICat,\TpSet}$ for the obvious restriction of
$B:\brk{\WICat,\TpSet}$ to the coslice:

\iblock{
  \mrow{
    A:\Ol{\TpSet} \vdash \TpRef{A} : \Ol{\TpSet}
  }
  \mrow{
    \TpRef*_w\prn{A : \brk{w\downarrow\WICat,\TpSet}}\, \prn{w'\geq w} \eqdef
    \Compr{
      i\in \vrt{w'}
    }{
      \Ltr\brk{B\leftarrow w'i}.
      B\Sub{\vert w'} = A\Sub{\vert w'}
    }
  }
}

We define the internal state monad and its $\Ltr$-algebra structure as follows,
using the weak bisimulation monad $\TpBisim*$ pointwise:

\iblock{
  \mrow{
    A:\Ol{\TpSet}\vdash \TpT{A} : \Ol{\TpSet}
  }
  \mrow{
    \TpT*_wA w_1 \eqdef
    \Forall{w_2\geq w_1}
    \HICat\Sub{w_2}
    \to
    \Alert{\TpBisim}
    \Exists{w_3\geq w_2}
    \HICat\Sub{w_3}\times
    A w_3
  }

  \row

  \mrow{
    A : \Ol{\TpSet} \vdash \vartheta\Sub{\TpT{A}} : \Ltr\TpT{A}\to\TpT{A}
  }
  \mrow{
    \vartheta\Sub{w,\TpT*_wA} m w' h \eqdef
    \vartheta\,\prn{
      \Next*\brk{u\leftarrow m}.
      u w' h
    }
  }

  \row

  \mrow{
    A:\Ol{\TpSet} \vdash \Con{ret}_A : A\to \TpT{A}
  }
  \mrow{
    \Con{ret}\Sub{w,A} a w' h \eqdef
    \eta\,\prn{
      \Con{pack}\,\prn{
        w',\prn{h,a\Sub{\vert{w_2}}}
      }
    }
  }

  \row

  \mrow{
    A,B:\Ol{\TpSet}\vdash \Con{bind}\Sub{A,B} : \TpT{A}\times\prn{A\to\TpT{B}}\to \TpT{B}
  }
  \mrow{
    \Con{bind}\Sub{w,A,B}\prn{m,k} w' h =
    \Con{pack}\,\prn{w'',\prn{h',a}} \leftarrow m w h;
    k_{w''} a w'' h'
  }
}

The steppers $\TmStep,\TmStep\SubL,\TmStep\SubR : \TpT{\TpUnit}$
are defined \emph{pointwise} using the corresponding constructs of
$\TpBisim*$.
We define the generic effects for the getter and setter below:

\iblock{
  \mrow{
    A:\Ol{\TpSet}\vdash \Con{get}_A : \TpRef{A}\to \TpT{A}
  }
  \mrow{
    \Con{get}\Sub{w,A}l w' h \eqdef
    \vartheta\,\prn{
      \Next*\brk{z\leftarrow hl}.\,
      \eta\,\prn{\Con{pack}\,\prn{w',\prn{h,z}}}
    }
  }

  \row

  \mrow{
    A:\Ol{\TpSet}\vdash \Con{set}_A : \TpRef{A}\times A\to\TpT{\TpUnit}
  }
  \mrow{
    \Con{set}\Sub{w,A}\prn{l,x}w' h \eqdef
    \eta\,\prn{
      \Con{pack}\,\prn{
        w',
        \prn{h\brk{l\mapsto x\Sub{\vert w'}},*}
      }
    }
  }
}

The allocator is the only subtle part; it is here that we will need to use the
strict extension structure that we have assumed on $\TpSet$ in $\ECat$.

\iblock{
  \mrow{
    A:\Ol{\TpSet}\vdash \Con{new}_A : A\to \TpT\prn{\TpRef{A}}
  }

  \mhang{
    \Con{new}\Sub{w,A}x w' h \eqdef
  }{
    \mrow{
      \Kwd{let}\ i = \Con{fresh}\,\vrt{w'};
    }
    \mrow{
      \Kwd{let}\ w'' = w' \cup\brc{i\mapsto \Next \Alert{\prn{\lambda w_0. \prn{w_0\geq w'}_*\prn{A w_0}}}};
    }
    \mrow{
      \Kwd{let}\ h' = h\Sub{\vert w''}\cup \brc{i\mapsto \Next{x\Sub{\vert{}w''}}};
    }
    \mrow{
      \eta\,\prn{
        \Con{pack}\,\prn{
          w'',
          \prn{
            i,
            h'
          }
        }
      }
    }
  }
}

The problem solved by strict extension above is the following: a type
$A : \Ol{\TpSet}$ at world $w$ has extent only on the coslice
$w\downarrow\WICat$, and yet to extend the world by a new cell we
must provide a type that has extent on all of $\WICat$. We use strict extension
to extend $A$ to a global type in a way that agrees \emph{strictly} with $A$
after $w'\geq w$.
We synthesize the preceding discussion into the following result:
\begin{thm}
  Let $\ECat$ be an elementary topos model of $\iGDTT$. The lifting of this
  model to $\brk{\WICat,\SPAN{\ECat}}$ extends to a model of \iGDTTRefLRAT{}.
\end{thm}
\begin{proof}
  Applying \cref{thm:igdtt-model-lifting} we obtain a model of $\iGDTT{}$. The
  relational constructs of \iGDTTRefLRAT{} ($\OpnL$, $\OpnR$, \etc{}) are lifted
  pointwise from $\SPAN{\ECat}$. The remaining constructs---those governing
  state---are interpreted using the above definitions of $\Con{T}$, $\Con{new}$,
  $\Con{get}$, and $\Con{set}$.
\end{proof}

Recalling that \iGDTTRef{} embeds into \iGDTTRefLRAT{}, we obtain the following
by the same reasoning as \cref{cor:igdtt-consistent}:
\begin{cor}
  Both \iGDTTRef{} and \iGDTTRefLRAT{} are consistent.
\end{cor}

Fix a closed term $M : \TpT{\mathbb{N}}$. Inspecting the interpretation of $M$
within $\brk{\WICat,\SPAN{\ECat}}$, we see that it corresponds to a global
element of $\Mor[\bbrk{M}]{1}{\TpT{\mathbb{N}}}$ in
$\brk{\WICat,\SPAN{\ECat}}$. As $\WICat$ possesses an initial object (the empty heap), $\bbrk{M}$ is determined by its instantiation
$\Mor[\bbrk{M}\Sub{\emptyset}]{1}{\Con{T}\Sub{\emptyset} \mathbb{N}}$.

Given any semantic world $w$ and heap $h : \HICat\Sub{w}$, we therefore may run
$t\Sub{\emptyset}$ to obtain an element $\TpBisim \mathbb{N}$ within
$\SPAN{\ECat}$. We present this using \iGDTT{} as the internal language of
$\SPAN{\ECat}$:

\iblock{
  \mrow{
    \Con{run} : \prn{w : \WICat} \to \HICat\Sub{w} \to \TpBisim \mathbb{N}
  }
  \mrow{
    \Con{run}_w\ h\ t = \prn{w', h', n} \leftarrow t\Sub{\emptyset}\ w\ h; \eta\prn{n}
  }
}

Denoting the restriction of $\bbrk{M}$ by $\OpnL$ (resp. $\OpnR$) by
$\bbrk{M}\SubL$ (resp. $\bbrk{M}\SubR$), the preceding discussion now
substantiates the claim made in \cref{sec:case-study:local-references}:

\begin{thm}
  Fix a closed term $M : \TpT{\TpNat}$ in \iGDTTRefLRAT{}. For any $w : \WICat$
  and a heap $h : \HICat\Sub{w}$, there is an element of
  $\Con{WeaklyBisimilar}\Sub{\mathbb{N}}{\brk{\OpnL\hookrightarrow \Con{run}_w\
      h\ \bbrk{M}\SubL, \OpnR\hookrightarrow \Con{run}_w\ h\ \bbrk{M}\SubR}}$ \ie{} a
  weak bisimulation between the results of executing $\bbrk{M}\SubL$ and
  $\bbrk{M}\SubR$.
\end{thm}

\section{Conclusions and future work}\label{sec:concl}

\paragraph{Dependent type theory with higher-order store}

We have defined denotational semantics for higher-order store and polymorphism
in impredicative guarded dependent type theory (\iGDTT{}); this denotational
semantics, in turn, justifies the extension of impredicative guarded dependent
type theory with general reference types and a (higher-order) state monad, a
theory that we call \iGDTTRef{}. Although denotational semantics for
higher-order state are \emph{a priori} desirable in general, they are
essentially mandatory in the case of dependent type theory --- an area in which
operational methods have proved disappointingly unscalable.

\paragraph{Relational reasoning atop an intensional equational theory}

The equational theory of store that we model is, however, quite intensional:
reads from the store leave behind ``abstract steps'' that cannot be ignored.
Thus in order to support reasoning about stateful computation, we have extended
\iGDTTRef{} with constructs for proof-relevant relational reasoning based on
the \DefEmph{logical relations as types} (LRAT) principle: in \iGDTTRefLRAT{},
it is possible to exhibit correspondences between stateful computations that
combine weak bisimulation for computation steps with parametricity for abstract
data types. Because we have ``full-spectrum'' dependent type theory at our
disposal, bisimulation arguments that in prior work required unfolding a very
complex operational or denotational semantics can be carried out totally na\"ively within the
type theory.

\subsection{Future work}

\subsubsection{Swapping and dropping allocations}

Our account of correspondences on worlds and stateful computations is somewhat
provisional; although we support storing relational invariants in the heap and
working ``up to'' computation steps, it is not possible in our model to exhibit
a correspondence between two programs that \emph{allocate} differently, even if
this does not affect their observable behavior. Two things are currently
missing from our model:
\begin{enumerate}

  \item \emph{Swapping of independent allocations.} To identify
    $\prn{\TmBind{r}{\TmNew{\mathbb{Z}}{5}}\TmBind{s}{\TmNew{\mathbb{Z}}{6}}
    u\, r\, s}$ with the reordered program
    $\prn{\TmBind{s}{\TmNew{\mathbb{Z}}{6}}\TmBind{r}{\TmNew{\mathbb{Z}}{5}}
    u\, r\, s}$, we would need the relational interpretation of worlds to allow
    the left-hand world to differ from the right-hand world by a permutation.

  \item \emph{Dropping of inactive allocations.} Ours is a model of
    \emph{global} state and allocation: thus the denotations of
    $\prn{\TmBind{r}{\TmNew{\mathbb{Z}}{5}} \TmRet{\prn{}}}$ and $\TmRet\prn{}$
    are distinct. One of our goals for future work is to extend our relational
    layer to account for such identifications, potentially by adapting the
    ideas of Kammar~\etal~\cite{kammar-levy-moss-staton:2017} who have themselves extended
    the classic Plotkin--Power account of local
    state~\cite{plotkin-power:2002} to support storage of pointers. We believe
    such an adaptation is plausible, but many questions remain unanswered in
    the space between storage of ground types and pointers with syntactic
    worlds and full thunk storage with semantic worlds.

\end{enumerate}

We believe that it may be possible to address the limitations outlined above
using \emph{nominal}
techniques~\cite{gabbay-pitts:2002,pitts:2013,pitts:2016}; it is also worth
considering the relationship between the \emph{logical relations as types}
principle and the proof-relevant logical relations for swapping and dropping
employed by Benton, Hofmann, and Nigam~\cite{benton-hofmann-nigam:2013}.

\subsubsection{More sophisticated Kripke worlds}

We also expect that it may be profitable to build on prior work in the
operational semantics of
state~\cite{ahmed-dreyer-rossberg:2009,dreyer-neis-rossberg-birkedal:2010}
involving yet more sophisticated Kripke worlds to facilitate local reasoning on
the heap. Such sophistication inevitably draws one in the direction of proper
program logics, however, which is another area of future work.

\subsubsection{Higher-order separation logic for modular reasoning about stateful programs}

Closely related to the question of Kripke worlds broached above, it is
reasonable to consider layering atop our type theory the higher-order
separation logic of Iris~\cite{iris:2015} to support modular reasoning about
stateful programs. Subsequent to circulation of the present paper in draft
form, Aagaard, Sterling, and Birkedal~\cite{aagaard-sterling-birkedal:2023}
have indeed defined exactly such a program logic for the denotational model of
System~\FMuRef{} given in \cref{sec:easy}, exhibiting a suitable
BI-hyperdoctrine over $\brk{\WICat,\VICat}$ internal to \iGDTT{}. It remains an
open question how to adapt this separation logic to the more sophisticated
model of dependent type theory with reference types given in
\cref{sec:semantics}; although logic-enriched dependent type theory is
well-understood~\cite{jacobs:1999,phoa:1992,taylor:1999}, we expect that
satisfactorily developing the theory of higher-order separation logic over
dependent types may require some new ideas.

\bibliographystyle{alphaurl}
\bibliography{references/refs-bibtex}

\end{document}